\documentclass[11pt]{article}
  
\usepackage[letterpaper,margin=1in]{geometry}
\usepackage{bbold}
\usepackage{amsmath}
\usepackage{amsthm}
\usepackage{amssymb}
\usepackage{pgffor}
\usepackage[inline]{enumitem}
\usepackage{mathtools}
\usepackage{bm}

\usepackage{times}

\usepackage{array}

\usepackage{url}
\usepackage{hyperref}

\hypersetup{colorlinks, linkcolor=darkblue, citecolor=darkgreen, urlcolor=darkblue}

\usepackage{float}
\usepackage{graphicx}

\graphicspath{{graphics/}}

\usepackage{tikz}

\usetikzlibrary{calc}
\usetikzlibrary{decorations.pathreplacing}
\usetikzlibrary{decorations.pathmorphing}
\usetikzlibrary{shapes.geometric}
\usetikzlibrary{fadings,decorations,fit}

\usepackage[vlined,ruled]{algorithm2e}

\usepackage[textsize=scriptsize]{todonotes}

\definecolor{darkblue}{rgb}{0,0,0.38}
\definecolor{darkred}{rgb}{0.6,0,0}
\definecolor{darkgreen}{rgb}{0.1,0.35,0}

\DeclareMathOperator{\spn}{span}
\DeclareMathOperator{\val}{val}
\DeclareMathOperator{\polylog}{polylog}

\makeatletter
\newcommand{\labeltarget}[1]{\Hy@raisedlink{\hypertarget{#1}{}}}
\makeatother

\newcommand\tm[3][]{\tikz[remember picture,baseline=(#2.base),#1]%
{\node[inner sep=0em] (#2) {#3};}}

\newcommand{\Pcomp}{\textsc{P}}
\newcommand{\NPcomp}{\textsc{NP}}

\newcommand{\Vl}{V^{\mathcal{L}}}
\newcommand{\Vt}{V^{\mathcal{T}}}

\def\supp{\operatorname{supp}}

\def\OPT{\mathsf{OPT}}

\newcommand{\fl}[1]{\lfloor #1 \rfloor}

\newtheorem{theorem}{Theorem}
\newtheorem{lemma}[theorem]{Lemma}

\newtheorem{corollary}[theorem]{Corollary}

\newtheorem{property}[theorem]{Property}

\title{Firefighting on Trees Beyond Integrality Gaps}

\author{
David Adjiashvili\thanks{
Department of Mathematics, ETH Zurich, Zurich, Switzerland.
Email: \href{mailto:david.adjiashvili@ifor.math.ethz.ch}%
{david.adjiashvili@ifor.math.ethz.ch}.
Supported by Seed Project ``Risk Protection in Complex Networks''
of ETH Zurich Risk Center.}
\and 
Andrea Baggio\thanks{
Department of Mathematics, ETH Zurich, Zurich, Switzerland.
Email: \href{mailto:andrea.baggio@ifor.math.ethz.ch}%
{andrea.baggio@ifor.math.ethz.ch}.
Supported by EU grant FP7-PEOPLE-2012-ITN no. 316647,
``Mixed-Integer Nonlinear Optimization''.
}
\and
Rico Zenklusen\thanks{
Department of Mathematics, ETH Zurich, Zurich, Switzerland.
Email: \href{mailto:ricoz@math.ethz.ch}%
{ricoz@math.ethz.ch}.}
}

\begin{document}

\maketitle

\begin{abstract}

The Firefighter problem and a variant of it, known as
Resource Minimization for Fire Containment (RMFC), are natural
models for optimal inhibition of harmful spreading processes.
Despite considerable progress on several
fronts, the approximability of these problems is still badly understood.
This is the case even when the underlying graph is a tree, which is
one of the most-studied graph structures in this context
and the focus of this paper.
In their simplest version, a fire spreads from one fixed vertex
step by step from burning to adjacent non-burning vertices,
and at each time step, $B$ many non-burning vertices
can be protected from catching fire.
The Firefighter problem asks, for a given $B$,
to maximize the number of vertices
that will not catch fire, whereas RMFC (on a tree)
asks to find the smallest $B$ that allows for saving all leaves
of the tree. 
Prior to this work, the best known approximation ratios
were an $O(1)$-approximation for the Firefighter problem
and an $O(\log^* n)$-approximation for RMFC, both being
LP-based and essentially matching
the integrality gaps of two natural LP relaxations.

We improve on both approximations by presenting a PTAS for the
Firefighter problem and an $O(1)$-approximation for RMFC,
both qualitatively matching the known hardness results.
Our results are obtained through a combination
of the known LPs with
several new techniques, which allow for
efficiently enumerating subsets of super-constant
size of a good solution to obtain stronger LPs.
\end{abstract}

\section{Introduction}

The Firefighter problem was introduced by
Hartnell~\cite{Hartnell1995} as a natural model for optimal
inhibition of harmful spreading phenomena on a graph.
Despite considerable interest in the problem
and progress on several fronts, our understanding
of how well this and related problems can be approximated
is still very limited.
Interestingly, this is even true when the underlying
graph is a spanning tree,
which is one of the most-studied graph
structures in this context and also the focus of
this paper.

The Firefighter problem on trees is defined as follows.
We are given a graph $G=(V,E)$ which is a spanning tree
and a vertex $r\in V$, called \emph{root}.
The problem is defined over discretized time steps.
At time $0$, a fire starts at $r$ and spreads step
by step to neighboring vertices.
During each time step $1,2,\dots$ an arbitrary non-burning
vertex $u$ can be \emph{protected}, preventing $u$
from burning in any future time step.
In its original form, the goal is to find a protection
strategy minimizing the number of vertices that will
catch fire.
A closely related problem, called \emph{Resource
Minimization for Fire Containment (RMFC)} on trees,
was introduced by
Chalermsook and Chuzhoy~\cite{ChalermsookChuzhoy2010}.
Here the task is to determine the
smallest number $B\in \mathbb{Z}_{>0}$
such that if one can protect $B$ vertices at each time
step (instead of just $1$), then there is a protection
strategy where none of the leaves of the tree
catches fire.
In this context, $B$ is referred to as
the \emph{number of firefighters}.

Both the Firefighter problem and RMFC---both restricted
to trees as defined above---are known to be computationally
hard problems.
More precisely, Finbow, King, MacGillivray and
Rizzi~\cite{FinbowKingMacGillivrayRizzi2007} showed 
NP-hardness for the Firefighter problem on trees with 
maximum degree three.
For RMFC on trees, it is NP-hard to decide
whether one firefighter,
i.e., $B=1$, is sufficient~\cite{king_2010_firefighter};
thus, unless $\Pcomp=\NPcomp$, there is
no (efficient) approximation algorithm with an approximation
factor strictly better than $2$.

On the positive side, several approximation algorithms
have been suggested for the Firefighter problem and RMFC.
Hartnell and Li~\cite{HartnellLi2000} showed that a 
natural greedy algorithm is a $\frac{1}{2}$-approximation for the
Firefighter problem. This approximation guarantee
was later improved by Cai, Verbin
and Yang~\cite{CaiVerbinYang2008} to $1-\frac{1}{e}$,
using a natural linear programming (LP) relaxation
and dependent randomized rounding.
It was later observed by Anshelevich, Chakrabarty, Hate 
and Swamy~\cite{AnshelevichChakrabartyHateSwamy2009}
that the Firefighter problem on
trees can be interpreted as a monotone submodular
function maximization (SFM) problem subject to a
partition matroid constraint.
This leads to alternative ways to obtain a
$(1-\frac{1}{e})$-approximation by using a recent
$(1-\frac{1}{e})$-approximation for monotone SFM subject
to a matroid
constraint~\cite{vondrak_2008_optimal,calinescu_2011_maximizing}.
The factor $1-\frac{1}{e}$ was later only improved for
various restricted tree
topologies (see~\cite{IwaikawaKamiyamaMatsui2011})
and hence, for arbitrary trees,
this is the best known approximation factor to date.

For RMFC on trees, Chalermsook and
Chuzhoy~\cite{ChalermsookChuzhoy2010} presented an
$O(\log^* n)$-approximation, where $n=|V|$ is the
number of vertices.\footnote{
$\log^* n$ denotes the minimum number $k$ of
logs of base two that have to be nested such
that $\underbrace{\log\log\dots\log}_{k \text{ logs}} n \leq 1$.}
Their algorithm is based on a natural
linear program which is a straightforward
adaptation of the one used in~\cite{CaiVerbinYang2008}
to get a $(1-\frac{1}{e})$-approximation for the
Firefighter problem on trees.

Whereas there are still considerable gaps between
current hardness results and approximation algorithms
for both the Firefighter problem and RMFC on trees, 
the currently best approximations essentially match
the integrality gaps of the underlying LPs.
More precisely,
Chalermsook and Vaz~\cite{chalermsook_2016_new}
showed that for any $\epsilon >0$, the canonical LP used
for the Firefighter problem on trees has an integrality
gap of $1-\frac{1}{e}+\epsilon$. This generalized
a previous result by Cai, Verbin and Yang~\cite{CaiVerbinYang2008},
who showed the same gap if the integral solution is required
to lie in the support of an optimal LP solution.
For RMFC on trees, the integrality gap
of the underlying LP
is~$\Theta(\log^* n)$~\cite{ChalermsookChuzhoy2010}.

It remained open to what extent these integrality
gaps may reflect the approximation hardnesses of the
problems.
This question is motivated by two related problems
whose hardnesses of approximation indeed matches the
above-mentioned integrality gaps for the Firefighter
problem and RMFC.
In particular, many versions of monotone SFM subject
to a matroid constraint---which we recall was shown
in~\cite{AnshelevichChakrabartyHateSwamy2009}
to capture the Firefigther problem on
trees as a special
case---are
hard to approximate up to a factor of
$1-1/e+\epsilon$ for any constant $\epsilon >0$.
This includes the problem of maximizing an explicitly
given coverage function subject to a single cardinality
constraint, as shown by Feige~\cite{feige_1998_threshold}.
Moreover, as highlighted in~\cite{ChalermsookChuzhoy2010},
the Asymmetric $k$-center problem is similar in nature
to RMFC, and has an approximation hardness of
$\Theta(\log^* n)$.

The goal of this paper is to fill the gap between
current approximation ratios and hardness results
for the Firefighter problem and RMFC on trees.
In particular, we present approximation ratios
that nearly match the hardness results, thus showing
that both problems can be approximated to factors
that are substantially better than the integrality
gaps of the natural LPs.
Our results are based on several new techniques,
which may be of independent interest.

\subsection{Our results}

Our main results show
that both the Firefighter
problem and RMFC admit strong approximations
that essentially match known hardness bounds,
showing that approximation factors can be achieved that
are substantially stronger than the integrality
gaps of the natural LPs.
In particular, we obtain the following result
for RMFC.
\begin{theorem}\label{thm:O1RMFC}
There is a $12$-approximation for RMFC.
\end{theorem}
Recalling that RMFC is hard to approximate
within any factor better than $2$, the above
result is optimal up to a constant factor,
and improves on the previously best
$O(\log^* n)$-approximation of
Chalermsook and Chuzhoy~\cite{ChalermsookChuzhoy2010}.

Moreover, our main result for the Firefighter problem
is the following, which, in view of NP-hardness 
of the problem, is essentially best possible in
terms of approximation guarantee.
\begin{theorem}\label{thm:PtasFF}
There is a PTAS for the Firefighter problem
on trees.\footnote{A polynomial
time approximation scheme (PTAS) is
an algorithm that, for any constant
$\epsilon > 0$, returns in polynomial time
a $(1-\epsilon)$-approximate solution.}
\end{theorem}

Notice that the Firefighter problem does not admit
an FPTAS\footnote{A fully polynomial time approximation
scheme (FPTAS) is a PTAS with running
time polynomial in the input size and $\frac{1}{\epsilon}$.}
unless $\Pcomp=\NPcomp$, since the optimal
value of any Firefighter problem on a tree of
$n$ vertices is bounded by $O(n)$.\footnote{
The nonexistence of FPTASs unless $\Pcomp=\NPcomp$
can often be derived easily from strong
NP-hardness. Notice that the Firefighter problem
is indeed strongly NP-hard because
its input size is $O(n)$, in which case NP-hardness is
equivalent to strong NP-hardness.
}
We introduce several new techniques that allow us
to obtain approximation factors well beyond 
the integrality gaps of the natural LPs,
which have been a barrier for previous approaches.
We start by providing an overview of these techniques.

\smallskip

Despite the fact that we obtain
approximation
factors beating the integrality gaps, the natural LPs
play a central role in our approaches.
We start by introducing general transformations
that allow for transforming the Firefighter problem
and RMFC into a more compact and better structured
form, only losing small factors in terms of
approximability.
These transformations by themselves do not decrease
the integrality gaps.
However, they allow us to identify small
substructures, over which we can
optimize efficiently, and having an optimal solution
to these subproblems we can define
a residual LP with small integrality gap.

Similar high-level approaches,
like guessing a constant-size
but important subset of an optimal solution are well-known
in various contexts to decrease 
integrality gaps of natural LPs. The best-known example
may be classic PTASs for the knapsack problem,  where the
integrality gap of the natural LP can be decreased to
an arbitrarily small constant by first guessing a constant
number of heaviest elements of an optimal solution.
However, our approach differs substantially 
from this standard enumeration idea.
Apart from the above-mentioned transformations which,
as we will show later, already lead to new results
for both RMFC and the Firefighter problem, we
will introduce new combinatorial approaches to gain information
about a \emph{super-constant} subset of an optimal solution.
In particular, for the RMFC problem we define
a recursive enumeration
algorithm which, despite being very slow for enumerating all
solutions, can be shown to reach a good subsolution 
within a small recursion depth that can be reached in
polynomial time.
This enumeration procedure explores the space step by
step, and at each step we first solve an LP that determines
how to continue the enumeration in the next step.
We think that this LP-guided enumeration
technique may be of independent interest.
For the Firefighter problem, we use a well-chosen
enumeration procedure to identify a polynomial
number of additional constraints to be added to the
LP, that improves its integrality gap to
$1-\epsilon$.

\subsection{Further related results}
Iwaikawa, Kamiyama and Matsui~\cite{IwaikawaKamiyamaMatsui2011}
showed that the approximation guarantee of $1-\frac{1}{e}$ can be improved 
for some restricted families of trees, in particular of low maximum degree.
Anshelevich, Chakrabarty, Hate and 
Swamy~\cite{AnshelevichChakrabartyHateSwamy2009} studied the approximability
of the Firefighter problem in general graphs, which they prove admits no 
$n^{1-\epsilon}$-approximation for any $\epsilon > 0$, unless
$\Pcomp=\NPcomp$. In a 
different model, where the protection also spreads through the graph 
(the \emph{Spreading Model}), the authors show that the problem 
admits a polynomial $(1-\frac{1}{e})$-approximation on general graphs. 
Moreover, for RMFC, an $O(\sqrt n)$-approximation for general 
graphs and an $O(\log n)$-approximation for directed layered
graphs is presented.
The
latter result was obtained independently by Chalermsook and Chuzhoy~\cite{ChalermsookChuzhoy2010}.
Klein, Levcopoulos and Lingas~\cite{KleinLevcopoulosLingas2014} introduced a 
geometric variant of the Firefighter problem, proved its NP-hardness and provided 
a constant-factor approximation algorithm. 
The Firefighter problem and RMFC are natural special cases
of the Maximum Coverage Problem with 
Group Constraints (MCGC)~\cite{ChekuriKumar2004} and the 
Multiple Set Cover problem (MSC)~\cite{ElkinKortsarz2006}, respectively. 
The input in MCGC is a set system consisting of a finite set $X$
of elements with nonnegative weights, a
collection of subsets $\mathcal{S} = \{S_1, \cdots, S_k\}$ of $X$ and 
an integer $k$. The sets in $\mathcal{S}$ are partitioned into 
groups $G_1, \cdots, G_l\subseteq \mathcal{S}$.
The goal is to pick a
subset $H\subseteq \mathcal{S}$ of $k$ 
sets from $\mathcal{S}$ whose union covers elements of total
weight as large as possible with the 
additional constraint that $|H \cap G_j| \leq 1$
for all $j\in [l]\coloneqq \{1,\dots, l\}$. In MSC,
instead of the fixed bounds for groups and the parameter $k$, the goal is
to choose a subset $H\subseteq \mathcal{S}$
that covers $X$ completely, while 
minimizing $\max_{j\in [l]}|H \cap G_j|$. The Firefighter
problem and RMFC can naturally be interpreted as special cases of the latter
problems with a laminar set system $\mathcal{S}$.

The Firefighter problem admits polynomial time algorithms in some restricted classes 
of graphs. Finbow, King, MacGillivray and Rizzi~\cite{FinbowKingMacGillivrayRizzi2007}
showed that, while the problem is NP-hard on trees with maximum degree three, 
when the fire starts at a vertex with degree two in a subcubic tree, the problem is solvable 
in polynomial time. Fomin, Heggernes and van Leeuwen~\cite{FominHeggernes_vanLeeuwen2012}
presented polynomial algorithms for interval graphs, split graphs,
permutation graphs and $P_k$-free graphs.

Several sub-exponential exact algorithms were developed for the Firefighter
problem on trees. Cai, Verbin and Yang~\cite{CaiVerbinYang2008} presented
a $2^{O(\sqrt{n}\log n)}$-time algorithm. 
Floderus, Lingas and Persson~\cite{FloderusLingasPersson2013} presented a
simpler algorithm with a slightly better running time, as well as a
sub-exponential algorithm for general graphs in the
spreading model and an $O(1)$-approximation in planar graphs
under some further conditions.

Additional directions of research on the Firefighter problem include parameterized
complexity (Cai, Verbin and Yang~\cite{CaiVerbinYang2008}, Bazgan, Chopin and 
Fellows~\cite{BazganChopinFellows2011}, Cygan, Fomin and 
van Leeuwen~\cite{CyganFomin_vanLeeuwen2012} and 
Bazgan, Chopin, Cygan, Fellows, Fomin and 
van Leeuwen~\cite{BazganChopinCyganFellowsFomin_van_Leeuwen2014}), generalizations
to the case of many initial fires and many firefighters (Bazgan, Chopin and 
Ries~\cite{BazganChopinRies2013} and Costa, Dantas, Dourado, Penso and 
Rautenbach~\cite{CostaDantasDouradoPensoRautenbach2013}),
and the study of potential strengthenings of the canonical LP for
the Firefighter problem on trees (Hartke~\cite{Hartke2006} and Chalermsook and Vaz~\cite{chalermsook_2016_new}).

Computing the \emph{Survivability} of a graph is 
a further problem closely related to Firefighting
that has attracted considerable attention
(see~\cite{CaiWang2009,CaiChengVerbinZhou2010,Pralat2013,Esperet_Van_Den_HeuvelMaffraySipma2013,Gordinowicz2013,KongZhangWang2014}).
For a graph $G$ and a parameter
$k\in \mathbb{Z}_{\geq 0}$, the $k$-survivability of $G$ 
is the average fraction of nodes that one can save with $k$ firefighters
in $G$, when the fire starts at a random node.

For further references we refer the reader to the survey of Finbow and 
MacGillivray~\cite{FinbowMacGillivray2009}.

\subsection{Organization of the paper}

We start by introducing the classic linear programming
relaxations for the Firefighter problem and RMFC
in Section~\ref{sec:preliminaries}.
Section~\ref{sec:overview} outlines our main
techniques and algorithms. Some 
proofs and additional discussion
are deferred to later sections, namely
Section~\ref{sec:proofsCompression}, providing
details on a compression technique that is
crucial for both our algorithms, Section~\ref{sec:proofsFF},
containing proofs for results related to the
Firefighter problem, and Section~\ref{sec:proofsRMFC},
containing proofs for results related to RMFC.
Finally, Appendix~\ref{apx:trans} contains some basic
reductions showing how to reduce different variations
of the Firefighter problem to each other.

\section{Classic LP relaxations and preliminaries}
\label{sec:preliminaries}

Interestingly, despite the fact that we
obtain approximation factors considerably
stronger than the known integrality gaps of the
natural LPs,
these LPs still play a central role in our approaches.
We thus start by introducing the natural LPs together with
some basic notation and terminology.

Let $L\in \mathbb{Z}_{\geq 0}$ be the \emph{depth} of
the tree, i.e., the largest
distance---in terms of number of edges---between $r$
and any other vertex in $G$.
Hence, after at most $L$ time steps, the fire
spreading process will halt.
For $\ell\in [L]:=\{1,\dots, L\}$, let
$V_\ell\subseteq V$ be the set of all vertices of
distance $\ell$ from $r$, which we call the
\emph{$\ell$-th level} of the instance.
For brevity, we use $V_{\leq \ell} = \cup_{k=1}^\ell V_k$,
and we define in the same spirit $V_{\geq \ell}$, 
$V_{< \ell}$, and $V_{> \ell}$.
Moreover, we denote by
$\Gamma \subseteq V$ the set of all leaves
of the tree, and for any $u\in V$, the set
$P_u\subseteq V\setminus \{r\}$ denotes
the set of all vertices on the unique $u$-$r$ path
except for the root $r$.

The relaxation for RMFC used in~\cite{ChalermsookChuzhoy2010}
is the following:
\begin{equation}\label{eq:lpRMFC}
\begin{array}{*2{>{\displaystyle}r}c*2{>{\displaystyle}l}}
\min & B & & \\
     & x(P_u) &\geq &1   &\forall u\in \Gamma \\
     & x(V_{\leq \ell}) &\leq &B\cdot \ell\hspace*{2em}
               &\forall \ell\in [L]\\
     & x &\in &\mathbb{R}_{\geq 0}^{V\setminus \{r\}},
\end{array}\tag{$\mathrm{LP_{RMFC}}$}
\end{equation}
where $x(U):=\sum_{u\in U} x(u)$ for any $U\subseteq V\setminus \{r\}$.
Indeed, if one enforces $x\in \{0,1\}^{V\setminus \{r\}}$
and $B\in \mathbb{Z}$
in the above relaxation, an exact description of RMFC is
obtained where $x$ is the characteristic vector of the
vertices to be protected and $B$ is the number of
firefighters:
The constraints $x(P_u)\geq 1$ for $u\in \Gamma$ enforce 
that for each leaf $u$, a vertex between $u$ and $r$
will be protected, which makes sure that $u$ will not
be reached by the fire;
moreover, the
constraints $x(V_{\leq \ell})\leq B\cdot \ell$
for $\ell\in [L]$ describe the vertex sets that can be
protected given $B$ firefighters per time step
(see~\cite{ChalermsookChuzhoy2010} for more details).
Also, as already highlighted in~\cite{ChalermsookChuzhoy2010},
there is an optimal solution to RMFC (and also to the Firefighter
problem), that protects with the firefighters available at
time step $\ell$ only vertices in $V_\ell$.
Hence, the above relaxation
can be transformed into one with same optimal objective value
by replacing the constraints
$x(V_{\leq \ell})\leq B\cdot \ell$ \;$\forall\ell\in [L]$
by the constraints
$x(V_\ell) \leq B$ \;$\forall \ell\in [L]$.

The natural LP relaxation for the Firefighter
problem, which leads to the previously best
$(1-1/e)$-approximation presented in~\cite{CaiVerbinYang2008},
is obtained analogously.
Due to higher generality, and even more importantly
to obtain more flexibility
in reductions to be defined later, we work on a slight
generalization of the Firefighter problem on trees,
extending it in two ways:
\begin{enumerate}[nosep,label=(\roman*)]
\item Weighted version: vertices $u\in V\setminus \{r\}$ have
weights $w(u)\in \mathbb{Z}_{\geq 0}$, and the goal
is to maximize the total weight of vertices not catching
fire.
In the classical Firefighter problem all weights are one.

\item General budgets/firefighters:
We allow for having a different number of
firefighters at each time step, say $B_\ell \in \mathbb{Z}_{>0}$
firefighters for time step $\ell\in [L]$.\footnote{Without
loss of generality we exclude $B_\ell=0$, since a level
with zero budget can be eliminated through a simple
contraction operation. For more details we refer to
the proof of Theorem~\ref{thm:compressionFF} which,
as a sub-step, eliminates zero-budget levels.
}
\end{enumerate}
Indeed, the above generalizations are mostly for convenience
of presentation, since general budgets can be reduced to
unit budgets (see Appendix~\ref{apx:trans} for a proof):
\begin{lemma}\label{lem:genBudgetsToUnit}
Any weighted Firefighter problem on trees with $n$ vertices
and general budgets can be transformed efficiently 
into an equivalent weighted Firefighter problem with
unit-budgets and $O(n^2)$ vertices.
\end{lemma}
We also show in Appendix~\ref{apx:trans} that
up to an arbitrarily small error in terms
of objective, any weighted Firefighter instance can be
reduced to a unit-weighted one.
In what follows, we always assume to deal with
a weighted Firefighter instance if not specified
otherwise. Regarding the budgets, we will be
explicit about whether we work with unit or
general budgets, since some techniques are easier
to explain in the unit-budget case, even though
it is equivalent to general budgets by
Lemma~\ref{lem:genBudgetsToUnit}.

\smallskip

An immediate extension of the LP relaxation
for the unit-weighted unit-budget Firefighter
problem used in~\cite{CaiVerbinYang2008}---which
is based on an IP formulation
presented in~\cite{macgillivray_2003_firefighter}---leads
to the
following LP relaxation 
for the weighted Firefighter problem
with general budgets. 
For $u\in V$, we denote by $T_u\subseteq V$
the set of all vertices in the subtree starting
at $u$ and including $u$, i.e., all vertices $v$
such that the unique $r$-$v$ path in $G$ contains $u$.

\begin{equation}\label{eq:lpFF}
\begin{array}{*2{>{\displaystyle}r}c*2{>{\displaystyle}l}}
\max & \sum_{u\in V\setminus \{r\}} x_u w(T_u) & & \\
     & x(P_u) &\leq &1   &\forall u\in \Gamma \\
     & x(V_{\leq \ell}) &\leq &\sum_{i=1}^\ell B_i\hspace*{2em}
               &\forall \ell\in [L]\\[1.5em]
     & x &\in &\mathbb{R}_{\geq 0}^{V\setminus \{r\}}.
\end{array}\tag{$\mathrm{LP_{FF}}$}
\end{equation}
The constraints
$x(P_u)\leq 1$ exclude redundancies, i.e., a vertex
$u$ is forbidden of being protected if another vertex above
it, on the $r$-$u$ path, is already protected. This
elimination of redundancies allows for writing the objective
function as shown above.

We recall that the integrality gap of~\ref{eq:lpRMFC}
was shown to be $\Theta(\log^* n)$~\cite{ChalermsookChuzhoy2010},
and the integrality gap of~\ref{eq:lpFF} is
asymptotically  $1-1/e$
(when $n\to \infty$)~\cite{chalermsook_2016_new}.

\smallskip

Throughout the paper, all logarithms are of base $2$ if
not indicated otherwise.
When using big-$O$ and related notations
(like $\Omega, \Theta, \ldots$), we will always
be explicit about the dependence on 
small error terms $\epsilon$---as used when talking
about $(1-\epsilon)$-approximations---and not consider
it to be part of the hidden constant.
To make statements where $\epsilon$ is part of the
hidden constant, we will use the notation
$O_{\epsilon}$ and likewise
$\Omega_{\epsilon}, \Theta_{\epsilon},\ldots$.

\section{Overview of techniques and algorithms}
\label{sec:overview}

In this section, we present our main technical
contributions and outline our algorithms.
We start by introducing a compression technique in
Section~\ref{subsec:compression} that works
for both RMFC and the Firefighter problem and allows
for transforming any instance to one on a tree with
only logarithmic depth.
One key property we achieve with compression,
is that we can later use (partial)
enumeration techniques with exponential running
time in the depth of the tree. 
However, compression on its own already leads to
interesting results. In particular, it allows
us to obtain a QPTAS for
the Firefighter problem, and a quasipolynomial time
$2$-approximation for RMFC.\footnote{The running time
of an algorithm is \emph{quasipolynomial} if it is
of the form $2^{\polylog(\langle \mathrm{input} \rangle)}$,
where $\langle \mathrm{input} \rangle$ is the input
size of the problem. A QPTAS is an algorithm that,
for any constant $\epsilon >0$, returns
a $(1-\epsilon)$-approximation in quasipolynomial
time.}
However, it seems highly
non-trivial to transform these quasipolynomial time
procedures to efficient ones.

To obtain the claimed results, we develop two 
(partial) enumeration methods to reduce the integrality
gap of the LP.
In Section~\ref{subsec:overviewFirefighter}, we provide
an overview of our PTAS for the Firefighter problem,
and Section~\ref{subsec:overviewRMFC} presents 
our $O(1)$-approximation for RMFC.

\subsection{Compression}\label{subsec:compression}

Compression is a technique that is applicable to both
the Firefighter problem and RMFC. It allows for reducing
the depth of the input tree at a very small loss in
the objective.
We start by discussing compression in the context of
the Firefighter problem. 

To reduce the depth of the tree, we will
first do a sequence of what we call \emph{down-pushes}.
Each down-push acts on two levels $\ell_1,\ell_2\in [L]$
with $\ell_1 < \ell_2$ of the tree, and moves the budget $B_{\ell_1}$
of level $\ell_1$ down to $\ell_2$, i.e., the new
budget of level $\ell_2$ will be $B_{\ell_1}+B_{\ell_2}$,
and the new budget of level $\ell_1$ will be $0$.
Clearly, down-pushes only restrict our options for
protecting vertices. However, we can show that
one can do a sequence of down-pushes such that
first, the optimal objective value of the new instance
is very close to the one of the original instance,
and second,
only $O(\log L)$ levels have non-zero budgets.
Finally, levels with $0$-budget can easily be removed
through a simple contraction operation, thus leading
to a new instance with only $O(\log L)$ depth.

Theorem~\ref{thm:compressionFF} below
formalizes our main compression
result for the Firefighter problem, which we state for
unit-budget Firefighter instances for simplicity.
Since Lemma~\ref{lem:genBudgetsToUnit}
implies that every general-budget
Firefighter instance with $n$ vertices can
be transformed into a unit-budget Firefighter instance
with $O(n^2)$ vertices---and thus $O(n^2)$
levels---Theorem~\ref{thm:compressionFF} can also be used to reduce
any Firefighter instance on $n$ vertices to one
with $O(\frac{\log n}{\delta})$
levels, by losing a factor of at most $1-\delta$
in terms of objective.

\begin{theorem}\label{thm:compressionFF}
Let $\mathcal{I}$ be a unit-budget Firefighter instance 
on a tree with depth $L$, and let $\delta\in (0,1)$.
Then one can efficiently construct a general budget
Firefighter instance $\overline{\mathcal{I}}$ with depth
$L'=O(\frac{\log L}{\delta})$, and such that
the following holds, where $\val(\OPT(\overline{\mathcal{I}}))$
and $\val(\OPT(\mathcal{I}))$ are the optimal values of
$\overline{\mathcal{I}}$ and $\mathcal{I}$, respectively.

\smallskip

\begin{enumerate}[nosep,label=(\roman*)]
\item
$\val(\OPT(\overline{\mathcal{I}}))
  \geq (1-\delta) \val(\OPT(\mathcal{I}))$, and

\item any solution to $\overline{\mathcal{I}}$ can be
transformed efficiently into a solution of $\mathcal{I}$
with same objective value.
\end{enumerate}
\end{theorem}

For RMFC we can use a very similar compression technique
leading to the following.

\begin{theorem}\label{thm:compressionRMFC}
Let $G=(V,E)$ be a rooted tree of depth $L$.
Then one can construct efficiently a rooted
tree $G'=(V',E')$ with $|V'|\leq |V|$ and
depth $L'=O(\log L)$, such that:

\smallskip

\begin{enumerate}[nosep,label=(\roman*)]
\item If the RMFC problem on $G$ has a 
solution with budget $B\in \mathbb{Z}_{> 0}$ at
each level, then the RMFC problem on $G'$
with non-uniform budgets, where level $\ell \geq 1$
has a budget of $B_\ell=2^{\ell} \cdot B$, has a solution.

\item Any solution to the RMFC problem on $G'$,
where level $\ell$ has budget $B_\ell=2^{\ell} \cdot B$,
can be transformed efficiently into an RMFC solution
for $G$ with budget $2B$.
\end{enumerate}

\end{theorem}

Interestingly, the above compression results already
allow us to obtain strong quasipolynomial approximation
algorithms for the Firefighter problem and RMFC,
using dynamic programming.
Consider for example the RMFC problem. We can first guess
the optimal budget $B$, which can be done efficiently
since $B\in \{1,\dots, n\}$. Consider now the instance
$G'$ claimed by Theorem~\ref{thm:compressionRMFC}
with budgets $B_\ell = 2^{\ell} B$.
By Theorem~\ref{thm:compressionRMFC},
this RMFC instance is feasible
and any solution to it can be converted to one of
the original RMFC problem with budget $2B$.
It is not hard to see that, for the fixed budgets $B_\ell$,
one can solve the RMFC problem on $G'$ in quasipolynomial
time using a bottom-up dynamic programming approach.
More precisely, starting with the leaves and moving
up to the root, we compute for each vertex $u\in V$
the following table. Consider a subset of the available
budgets, which can be represented as a vector
$q\in [B_1]\times \dots \times [B_{L'}]$. For each such
vector $q$ we want to know whether or not using the sub-budget
described by $q$ allows for disconnecting $u$ from all
leaves below it.
Since $L'=O(\log L)$ and the size of each budget $B_\ell$
is at most the number of vertices, the table size is
quasipolynomial.
Moreover, one can check that these tables can be
constructed bottom-up in quasipolynomial time.
Hence, this approach leads to a quasipolynomial time
$2$-approximation for RMFC. We recall that there is no
efficient approximation algorithm with an approximation
ratio strictly below $2$,
unless $\Pcomp=\NPcomp$.
A similar dynamic programming approach for the Firefighter
problem on a compressed instance leads to a
QPTAS.

However, our focus is on efficient algorithms, and it
seems non-trivial to transform the
above quasipolynomial time dynamic programming approaches
into efficient procedures. To obtain our results,
we therefore combine
the above compression techniques with
further approaches to be discussed next.

\subsection{Overview of PTAS for Firefighter problem}
\label{subsec:overviewFirefighter}

Despite the fact that~\ref{eq:lpFF} has a large integrality
gap
---which can be shown to be the case even after
compression\footnote{This follows from the fact
that through compression with some
parameter $\delta\in (0,1)$, both the optimal
value and optimal LP value change at most
by a $\delta$-fraction.}---%
it is a crucial tool in our PTAS.
Consider a general-budget
Firefighter instance,
and let $x$ be a vertex solution to~\ref{eq:lpFF}.
We say that a vertex $u\in V\setminus \{r\}$ is
\emph{$x$-loose}, or simply \emph{loose}, if
$u\in \supp(x):=\{v\in V\setminus \{r\} \mid x(v) > 0\}$
and $x(P_u) < 1$.
Analogously, we call a vertex $u\in V\setminus \{r\}$
\emph{$x$-tight}, or simply \emph{tight}, if
$u\in \supp(x)$ and $x(P_u)=1$.
Hence, $\supp(x)$ can be partitioned into
$\supp(x)=\Vl \cup \Vt$,
where $\Vl$ and $\Vt$
are the set of all loose and tight vertices, respectively.
Using a sparsity argument for vertex solutions
of~\ref{eq:lpFF} we can bound the number of
$x$-loose vertices.

\begin{lemma}\label{lem:sparsityFF}
Let $x$ be a vertex solution to~\ref{eq:lpFF}
for a Firefighter problem with general budgets.
Then the number of $x$-loose vertices is at
most $L$, the depth of the tree.
\end{lemma}

Having a vertex solution $x$ to~\ref{eq:lpFF},
we can consider a simplified LP obtained from~\ref{eq:lpFF}
by only allowing to protect vertices that are $x$-tight.
A simple yet useful property of $x$-tight vertices is that
for any $u,v\in \Vt$ with $u\neq v$
we have $u\not\in P_v$. Indeed, if $u\in P_v$, then
$x(P_u) \leq x(P_v) - x(v) < x(P_v)=1$ because $x(v)>0$.
Hence, no two tight vertices lie on the same leaf-root path.
Thus, when restricting~\ref{eq:lpFF} to $\Vt$,
the path constraints $x(P_u) \leq 1$ for $u\in \Gamma$ transform
into trivial constraints requiring $x(v)\leq 1$ for
$v\in \Vt$, and one can easily observe that the
resulting constraint system is totally unimodular because
it describes a laminar matroid constraint given by the
budget constraints (see~\cite[Volume B]{Schrijver2003} for more details on
matroid optimization).
Re-optimizing over this LP we get an integral solution
of objective value at least
$\sum_{u\in V\setminus \{r\}} x_u w(T_u)
  - \sum_{u\in \Vl} x_u w(T_u)$,
because the restriction of $x$ to $\Vt$
is still feasible for the new LP.

In particular, if $\sum_{u\in \Vl} x_u w(T_u)$
was at most $\epsilon\cdot\val(\OPT)$, where
$\val(\OPT)$ is the optimal value of the instance,
then this would lead to a PTAS.
Clearly, this is not true in general, since it would contradict
the $(1-\frac{1}{e})$-integrality gap of~\ref{eq:lpFF}.
In the following, we will present techniques to limit
the loss in terms of LP-value when re-optimizing
only over variables corresponding to tight vertices $\Vt$.

Notice that when we work with a compressed instance, by
first invoking Theorem~\ref{thm:compressionFF} with $\delta=\epsilon$,
we have $|\Vl|=O(\frac{\log N}{\epsilon})$, where $N$ is the number of
vertices in the original instance. Hence, a PTAS would
be achieved if for all $u\in \Vl$, we had
$w(T_u) = \Theta(\frac{\epsilon^2}{\log N})\cdot \val(\OPT)$.
One way to achieve this in quasipolynomial time is
to first guess a subset of $\Theta(\frac{\log N}{\epsilon^2})$ many
vertices of an optimal solution with highest impact,
i.e., among all vertices $u\in \OPT$ we
guess those with largest $w(T_u)$.
This techniques has been used in various other settings
(see for example~\cite{ravi_1996_constrained,grandoni_2014_new}
for further details)
 and leads to another QPTAS for the Firefighter problem.
Again, it is unclear how 
this QPTAS could be turned into an efficient procedure.

The above discussion motivates to investigate vertices
$u\in V\setminus \{r\}$ with 
$w(T_u) \geq \eta$ for some
$\eta = \Theta(\frac{\epsilon^2}{\log N}) \val(\OPT)$.
We call such vertices \emph{heavy};
later, we will provide an explicit definition of $\eta$
that does not depend on the unknown $\val(\OPT)$ and
is explicit about the hidden constant.
Let $H=\{u\in V\setminus \{r\} \mid w(u) \geq \eta\}$
be the set of all heavy vertices.
Observe that $G[H\cup \{r\}]$---i.e., the induced subgraph
of $G$ over the vertices $H\cup\{r\}$---is a subtree of
$G$, which we call the \emph{heavy tree}.

Recall that by the above discussion, if we work on a
compressed instance with $L=O(\frac{\log N}{\epsilon})$ levels,
and if an optimal vertex solution to~\ref{eq:lpFF}
has no loose vertices that are heavy, then an integral
solution can be obtained of value at 
least $1-\epsilon$ times the LP value.
Hence, if we were able to guess the
heavy vertices contained in an optimal solution,
the integrality gap of the reduced problem
would be small
since no heavy vertices are left in the LP,
and can thus not be loose anymore.

Whereas there are too many options
to enumerate over all possible subsets
of heavy vertices that an optimal solution
may contain, we will do a coarser
enumeration.
More precisely, we will partition
the heavy vertices into $O_{\epsilon}(\log N)$
subpaths and guess for each subpath whether
it contains a vertex of $\OPT$.
For this to work out we need
that the heavy tree has a very
simple topology;
in particular, it should only have
$O_{\epsilon}(\log N)$ leaves.
Whereas this does not hold in general,
we can enforce it by a further transformation
making sure that $\OPT$ saves a
constant-fraction of $w(V)$ which---as we
will observe next---indeed limits the
number of leaves of the heavy tree to $O_{\epsilon}(\log N)$.
Furthermore, this transformation is useful to
complete our definition of heavy vertices
by explicitly defining the threshold $\eta$.

\begin{lemma}\label{lem:pruning}
Let $\mathcal{I}$ be a general-budget Firefighter instance
on a tree $G=(V,E)$ with weights $w$.
Then for any $\lambda\in \mathbb{Z}_{\geq 1}$,
one can efficiently construct
a new Firefighter instance
$\overline{\mathcal{I}}$ on a subtree $G'=(V',E')$ of $G$
with same budgets,
by starting from $\mathcal{I}$ and applying
node deletions and weight reductions, such that
\smallskip
\begin{enumerate}[nosep, label=(\roman*)]
 \item\label{item:pruningSmallLoss}
    $\val(\OPT(\overline{\mathcal{I}})) \geq
    \left(1 - \frac{1}{\lambda}\right)
     \val(\OPT(\mathcal{I}))$, and

 \item\label{item:pruningLargeOpt}
   $\val(\OPT(\overline{\mathcal{I}})) \geq
         \frac{1}{\lambda} w'(V')$,
where $w'\leq w$ are the vertex weights
in instance $\overline{\mathcal{I}}$.
\end{enumerate}
\smallskip
The deletion of $u\in V$ corresponds to removing the whole
subtree below $u$ from $G$, i.e., all vertices in $T_u$.
\end{lemma}

Since Lemma~\ref{lem:pruning} constructs a new instance
using only node deletions and weight reductions, any
solution to the new instance is also a solution to the
original instance of at least the same objective value.

Our PTAS for the Firefighter problem first applies the
compression Theorem~\ref{thm:compressionFF} with $\delta=\epsilon/3$
and then Lemma~\ref{lem:pruning} with
$\lambda = \lceil\frac{3}{\epsilon} \rceil$ to obtain
a general budget Firefighter instance on a tree $G=(V,E)$.
We summarize the properties of this new instance $G=(V,E)$
below. As before, to avoid confusion, we denote by $N$
the number of vertices of the 
original instance.

\begin{property}\leavevmode
\label{prop:preprocessedFF}

\begin{enumerate}[nosep, label=(\roman*)]
\item The depth $L$ of $G$ satisfies $L=O(\frac{\log N}{\epsilon})$.

\item $\val(\OPT) \geq \lceil\frac{3}{\epsilon} \rceil^{-1} w(V)
  \geq \frac{1}{4}\epsilon w(V)$.

\item The optimal value $\val(\OPT)$ of the new instance is
at least a $(1-\frac{2}{3}\epsilon)$-fraction of the
optimal value of the original instance.

\item Any solution to the new instance can be transformed
efficiently into a solution of the original instance
of at least the same value.
\end{enumerate}
\end{property}

Hence, to obtain a PTAS for the original instance, it
suffices to obtain, for any $\epsilon >0$, a 
$(1-\frac{\epsilon}{3})$-approximation for an instance
satisfying Property~\ref{prop:preprocessedFF}.
In what follows, we assume to work with an instance
satisfying Property~\ref{prop:preprocessedFF} and show
that this is possible.

Due to the lower bound on $\val(\OPT)$ provided
by Property~\ref{prop:preprocessedFF}, we now define
the threshold
$\eta = \Theta(\frac{\epsilon}{\log N}) \val(\OPT)$
in terms of $w(V)$ by
\begin{equation*}
\eta = \frac{1}{12} \frac{\epsilon^2}{L} w(V),
\end{equation*}
which implies that we can afford losing $L$ times a weight
of $\eta$, which will sum up to a total loss of at most
$\frac{1}{12}\epsilon^2 w(V) \leq \frac{1}{3} \epsilon \val(\OPT)$,
where the inequality is due to
Property~\ref{prop:preprocessedFF}.

Consider again the heavy tree $G[H\cup \{r\}]$. Due to
Property~\ref{prop:preprocessedFF} its topology is quite
simple. More precisely, the heavy tree has
only $O(\frac{\log N}{\epsilon^3})$ leaves.
Indeed, each leaf $u\in H$ of the heavy tree fulfills
$w(T_u) \geq \eta$, and two different leaves $u_1,u_2\in H$
satisfy $T_{u_1} \cap T_{u_2} = \emptyset$; since the
total weight of the tree is $w(V)$, the heavy tree
has at most
$w(V)/\eta = 12 L / \epsilon^2 = O(\frac{\log N}{\epsilon^3})$
many leaves.

In the next step, we define a well-chosen
small subset $Q$ of heavy vertices
whose removal (together with $r$) from $G$ will
break $G$ into components of weight at most $\eta$.
Simultaneously, we choose $Q$ such that removing
it together with $r$ from the heavy tree breaks
it into paths, over which we will do an enumeration
later.

\begin{lemma}\label{lem:setQ}
One can efficiently determine a set $Q\subseteq H$
satisfying the following.

\begin{enumerate}[nosep,label=(\roman*)]
\item $|Q|=O(\frac{\log N}{\epsilon^3})$.
\item $Q$ contains all leaves and all vertices
of degree at least $3$ of the heavy tree,
except for the root $r$.
\item Removing $Q\cup\{r\}$ from $G$ leads to a
graph $G[V\setminus (Q\cup \{r\})]$ where each
connected component has vertices whose weight
sums up to at most $\eta$.
\end{enumerate}

\end{lemma}

For each vertex $q\in Q$, let $H_q\subseteq H$ be
all vertices that are visited when traversing the
path $P_q$ from $q$ to $r$
until (but not including) the next
vertex in $Q\cup \{r\}$.
Hence, $H_q$ is a subpath of the heavy tree such
that $H_q\cap Q = \{q\}$, which we call for brevity
a \emph{$Q$-path}.
Moreover the set of all $Q$-paths partitions $H$.

We use an enumeration procedure to determine on which
$Q$-paths to protect a vertex. Since $Q$-paths are subpaths
of leaf-root paths, we can assume that at most one vertex
is protected in each $Q$-path.
Our algorithm enumerates over all $2^{|Q|}$ possible
subsets $Z \subseteq Q$, where $Z$
represents the $Q$-paths on which we will protect a
vertex. Incorporating this guess into~\ref{eq:lpFF},
we get the following linear program~\ref{eq:lpFFZ}:

\begin{equation}\label{eq:lpFFZ}
\begin{array}{*2{>{\displaystyle}r}c*2{>{\displaystyle}l}}
\max & \sum_{u\in V\setminus \{r\}} x_u w(T_u) & & \\
     & x(P_u) &\leq &1   &\forall u\in \Gamma \\
     & x(V_{\leq \ell}) &\leq &\sum_{i=1}^\ell B_i\hspace*{2em}
               &\forall \ell\in [L]\\
     & x(H_q) &= &1 &\forall q\in Z\\
     & x(H_q) &= &0 &\forall q\in Q\setminus Z\\
     & x &\in &\mathbb{R}_{\geq 0}^{V\setminus \{r\}}.
\end{array}\tag{$\mathrm{LP_{FF}}(Z)$}\labeltarget{eq:lpFFZtarget}
\end{equation}

We start with a simple observation regarding~\ref{eq:lpFFZ}.
\begin{lemma}\label{lem:isFaceFF}
The polytope over which~\ref{eq:lpFFZ} optimizes is a face
of the polytope describing the feasible
region of~\ref{eq:lpFF}.
Consequently, any vertex solution of~\ref{eq:lpFFZ} is a
vertex solution of~\ref{eq:lpFF}.
\end{lemma}
\begin{proof}
The statement immediately follows by observing that
for any $q\in Q$, the inequalities $x(H_q)\leq 1$
and $x(H_q)\geq 0$ are valid inequalities
for~\ref{eq:lpFF}.
Notice that $x(H_q)\leq 1$ is a valid inequality
for~\ref{eq:lpFF} because $H_q$ is a subpath of
a leaf-root path, and the load on any leaf-root
path is limited to $1$ in~\ref{eq:lpFF}.
\end{proof}

Analogously to~\ref{eq:lpFF} we define loose and tight
vertices for a solution to~\ref{eq:lpFFZ}.
A crucial implication of Lemma~\ref{lem:isFaceFF} is that
Lemma~\ref{lem:sparsityFF} also applies to any vertex
solution of~\ref{eq:lpFFZ}.

We will show in the following that for any choice of
$Z\subseteq Q$, the integrality gap of~\ref{eq:lpFFZ}
is small and we can efficiently obtain an integral solution of
nearly the same value as the optimal value of~\ref{eq:lpFFZ}.
Our PTAS then follows by enumerating all $Z\subseteq Q$
and considering the set $Z\subseteq Q$ of 
all $Q$-paths on which $\OPT$ protects a vertex.
The low integrality gap of~\ref{eq:lpFFZ} will follow from the fact
that we can now limit the impact of loose vertices.
More precisely, any loose vertex outside of the heavy
tree has LP contribution at most $\eta$ by definition of
the heavy tree. Furthermore, for each loose vertex $u$
on the heavy tree, which lies on some $Q$-path $H_q$,
its load $x(u)$ can be moved to the single tight vertex
on $H_q$. As we will show, such a load redistribution
will decrease the LP-value by at most $\eta$, due to our choice of $Q$.

We are now ready to state our $(1-\frac{\epsilon}{3})$-approximation
for an instance satisfying Property~\ref{prop:preprocessedFF},
which, as discussed, implies a PTAS for the Firefighter problem.
Algorithm~\ref{alg:FF} describes our
$(1-\frac{\epsilon}{3})$-approximation.

\begin{algorithm}

\begin{enumerate}[rightmargin=1em]
\item Determine heavy vertices $H=\{u\in V \mid w(T_u) \geq \eta\}$,
where $\eta=\frac{1}{12} \frac{\epsilon^2}{L} w(V)$.

\item Compute $Q\subseteq H$ using Lemma~\ref{lem:setQ}.

\item For each $Z\subseteq Q$, obtain an optimal vertex solution
to~\ref{eq:lpFFZ}. Let $Z^*\subseteq Q$ be a set for which the
optimal value
of~\hyperlink{eq:lpFFZtarget}{$\mathrm{LP_{FF}(Z^*)}$}
is largest among
all subsets of $Q$, and let $x$ be an optimal vertex
solution to~\hyperlink{eq:lpFFZtarget}{$\mathrm{LP_{FF}(Z^*)}$}.

\item\label{algitem:reoptFF}
Let $V^{\mathcal{T}}$ be the $x$-tight vertices.
Obtain an optimal vertex solution
to~\ref{eq:lpFF} restricted to variables
corresponding to vertices in $V^{\mathcal{T}}$.
The solution will be a $\{0,1\}$-vector, being the
characteristic vector of a
set $U\subseteq V^{\mathcal{T}}$ which we return.
\end{enumerate}

\caption{A $(1-\frac{\epsilon}{3})$-approximation for
a general-budget Firefighter instance satisfying
Property~\ref{prop:preprocessedFF}.}
\label{alg:FF}
\end{algorithm}

The following statement completes the proof of
Theorem~\ref{thm:PtasFF}.
\begin{theorem}\label{thm:PtasFFProp}
For any general-budget Firefighter instance satisfying
Property~\ref{prop:preprocessedFF},
Algorithm~\ref{alg:FF} computes efficiently a feasible
set of vertices $U\subseteq V\setminus \{r\}$ to protect
that is a $(1-\frac{\epsilon}{3})$-approximation. 
\end{theorem}
\begin{proof}
First observe that the linear program solved in
step~\ref{algitem:reoptFF} will indeed lead to
a characteristic vector with only $\{0,1\}$-components.
This is the case since no two $x$-tight vertices
can lie on the same leaf-root path. Hence, as discussed
previously, the linear program~\ref{eq:lpFF} restricted
to variables corresponding to $V^{\mathcal{T}}$ is totally
unimodular; indeed, the leaf-root path constraints $x(P_u)\leq 1$
for $u\in \Gamma$
reduce to $x(v)\leq 1$ for $v\in V^{\mathcal{T}}$, and
the remaining LP corresponds to a linear program over a laminar
matroid, reflecting the budget constraints.
Moreover, the set $U$ is clearly budget-feasible since 
the budget constraints are enforced by~\ref{eq:lpFF}.
Also, Algorithm~\ref{alg:FF} runs in polynomial time
because $|Q|=O(\frac{\log N}{\epsilon^3})$
by Lemma~\ref{lem:setQ} and hence,
the number of subsets of $Q$ is bounded by
$N^{O(\frac{1}{\epsilon^3})}$.

It remains to show that $U$ is a
$(1-\frac{\epsilon}{3})$-approximation.
Let $\OPT$ be an optimal solution to the considered
Firefighter instance with value $\val(\OPT)$.
Observe first that the value $\nu^*$
of~\hyperlink{eq:lpFFZtarget}{$\mathrm{LP_{FF}(Z^*)}$}
satisfies $\nu^* \geq \val(\OPT)$, because
one of the sets $Z\subseteq Q$ corresponds to
$\OPT$, namely $Z=\{q\in Q \mid H_q\cap \OPT \neq \emptyset\}$,
and for this $Z$ the characteristic vector
$\chi^{\OPT}\in \{0,1\}^{V\setminus \{r\}}$
of $\OPT$ is feasible
for~\ref{eq:lpFFZ}.
We complete the proof of Theorem~\ref{thm:PtasFFProp}
by showing that the value $\val(U)$ of $U$ satisfies
$\val(U) \geq (1-\frac{\epsilon}{3}) \nu^*$.
For this we show how to transform an optimal solution
$x$ of~\hyperlink{eq:lpFFZtarget}{$\mathrm{LP_{FF}(Z^*)}$}
into a solution $y$
to~\hyperlink{eq:lpFFZtarget}{$\mathrm{LP_{FF}(Z^*)}$}
with $\supp(y) \subseteq V^{\mathcal{T}}$
and such that the objective value $\val(y)$ of $y$ satisfies
$\val(y)\geq (1-\frac{\epsilon}{3}) \nu^*$.

Let $V^{\mathcal{L}} \subseteq \supp(x)$ be the set of
$x$-loose vertices, and let $H$ be all heavy vertices,
as usual. To obtain $y$, we start with $y=x$
and first set $y(u)=0$ for each $u\in V^{\mathcal{L}}\setminus H$.
Moreover, for each $u\in V^{\mathcal{L}}\cap H$ 
we do the following. Being part of the heavy vertices and
fulfilling $x(u)>0$, the vertex $u$
lies on some $Q$-path $H_{q_u}$ for some $q_u\in Z^*$.
Because $x(H_{q_u})=1$, there is a tight vertex
$v\in H_{q_u}$. We move the $y$-value from vertex
$u$ to vertex $v$, i.e., $y(v) = y(v)+y(u)$ and
$y(u)=0$. This finishes the construction of $y$.
Notice that $y$ is feasible
for~\hyperlink{eq:lpFFZtarget}{$\mathrm{LP_{FF}(Z^*)}$},
because it was obtained from $x$ by reducing values
and moving values to lower levels. 

To upper bound the reduction of the LP-value when
transforming $x$ into $y$, we show that the modification
done for each loose vertex $u\in V^{\mathcal{L}}$ decreased
the LP-value by at most $\eta$.
Clearly, for each $u\in V^{\mathcal{L}}\setminus H$,
since $u$ is not heavy we have $w(T_u)\leq \eta$; thus
setting $y(u)=0$ will have an impact of at most $\eta$
on the LP value.
Similarly, for $u\in V^{\mathcal{L}}\cap H$, moving the
$y$-value of $u$ to $q_u$ decreases the LP objective value
by
\begin{equation*}
y(u) \cdot \left(w(T_u) - w(T_{v})\right)
\leq
w(T_u) - w(T_v)
= w(T_u \setminus T_{v})
\leq \eta,
\end{equation*}
where the last inequality follows by observing
that $T_u \setminus T_{v}\subseteq T_u\setminus T_{q_u}$
are vertices in the
same connected component of $G[V\setminus (Q\cup \{r\})]$,
and thus have a total weight of at most $\eta$
by Lemma~\ref{lem:setQ}.

Hence,
$\val(x) - \val(y) \leq |V^{\mathcal{L}}|
  \cdot \eta \leq L\cdot \eta$,
where the second inequality follows by
Property~\ref{prop:preprocessedFF}.
This completes the proof by 
observing 
that $|V^{\mathcal{L}}| \leq L$
by Lemma~\ref{lem:sparsityFF}, and thus
\begin{align*}
\val(y) &= \val(x) + \left(\val(y) - \val(x)\right)
\geq \val(\OPT) + \val(y) - \val(x)
\geq \val(\OPT) - L\cdot \eta\\
&= \val(\OPT) - \frac{1}{12}\epsilon^2 w(V)
\geq \left(1-\frac{1}{3}\epsilon\right) \val(\OPT),
\end{align*}
where the last inequality
is due to Property~\ref{prop:preprocessedFF}.

\end{proof}

\subsection{Overview of $O(1)$-approximation for RMFC}
\label{subsec:overviewRMFC}

Also our $O(1)$-approximation for RMFC uses the natural LP,
i.e,~\ref{eq:lpRMFC}, as a crucial tool to guide the algorithm.
Throughout this section we will work on a compressed instance
$G=(V,E)$ of RMFC, obtained through Theorem~\ref{thm:compressionRMFC}.
Hence, the number of levels is $L=O(\log N)$, where $N$ is the
number of vertices of the original instance. Furthermore, the
budget on level $\ell\in [L]$ is given by $B_\ell = 2^{\ell} B$.
The advantage of working with a compressed instance for
RMFC is twofold.
First, we will again apply sparsity reasonings to limit in certain
settings the number of loose (badly structured) vertices by the
number of levels of the instance.
Second, the fact that low levels---i.e., levels far away from
the root---have high budget, will allow
us to protect a large number of loose vertices by only
increasing $B$ by a constant.

For simplicity, we work with a slight variation 
of~\ref{eq:lpRMFC}, where we replace, for $\ell\in [L]$,
the budget constraints
$x(V_{\leq \ell}) \leq \sum_{i=1}^{\ell} B_i$
by $x(V_\ell) \leq B_\ell$.
For brevity, we define
\begin{equation*}
P_B = \left\{x\in \mathbb{R}_{\geq 0}^{V\setminus \{r\}}
    \;\middle\vert\;
  x(V_\ell) \leq B\cdot 2^\ell \;\;\forall \ell\in [L]
   \right\}.
\end{equation*}
As previously mentioned (and shown
in~\cite{ChalermsookChuzhoy2010}), the resulting LP
is equivalent to~\ref{eq:lpRMFC}.
Furthermore, since the budget $B$ for a feasible RMFC solution
has to be chosen integral, we require $B\geq 1$.
Hence, the resulting linear relaxation asks to find
the minimum $B\geq 1$ such that 
the following polytope is non-empty:
\begin{equation*}
\bar{P}_B = P_B \cap
  \left\{x\in \mathbb{R}^{V\setminus \{r\}}_{\geq 0}
\;\middle\vert\;
x(P_u)\geq 1 \;\;\forall u\in \Gamma\right\}.
\end{equation*}

We start by discussing approaches to partially round a
fractional point $x\in \bar{P}_B$, for some fixed budget $B\geq 1$.
Any leaf $u\in \Gamma$ is fractionally cut off from
the root through the $x$-values on $P_u$. A crucial property
we derive and exploit is that leaves that are 
(fractionally) cut off from $r$ largely on low levels,
i.e., there is high $x$-value on $P_u$ on vertices
far away from the root, can be cut off from the root
via a set of vertices to be protected that are budget-feasible
when increasing $B$ only by a constant.
To exemplify the above statement, consider the level
$h=\lfloor \log L \rfloor$ as a threshold to define
top levels $V_\ell$ as those with indices $\ell\leq h$
and bottom levels when $\ell > h$. For any leaf
$u \in \Gamma$,
we partition the path $P_u$ into its top
part $P_u \cap V_{\leq h}$ and its bottom part
$P_u \cap V_{> h}$. Consider all leaves that are cut
off in bottom levels by at least $0.5$ units:
$W=\{u\in \Gamma \mid x(P_u\cap V_{> h}) \geq 0.5\}$.
We will show that there is a subset of
vertices $R\subseteq V_{>h}$ on bottom levels
to be protected that
is feasible for budget $\bar{B}=2B+1 \leq 3B$ and cuts off
all leaves in $W$ from the root.
We provide a brief sketch why this result holds,
and present a formal proof later.
If we set all entries of $x$ on top levels $V_{\leq h}$
to zero, we get a vector $y$ with $\supp(y) \subseteq V_{>h}$
such that $y(P_u) \geq 0.5$ for $u\in W$. Hence, $2y$ fractionally
cuts off all vertices in $W$ from the root and is feasible
for budget $2B$. To increase sparsity, we can replace $2y$ by
a vertex $\bar{z}$ of the polytope
\begin{equation*}
Q=\left\{z\in \mathbb{R}_{\geq 0}^{V\setminus \{r\}}
 \;\middle\vert\;
 z(V_\ell) \leq 2B\cdot 2^\ell \;\;\forall \ell\in [L],
 z(V_{\leq h}) = 0, z(P_u)\geq 1 \;\;\forall u\in W\right\},
\end{equation*}
which describes possible ways to cut off $W$ from $r$
only using levels $V_{> h}$, and $Q$ is non-empty
since $2y\in Q$.
Exhibiting a sparsity reasoning analogous to the
one used for the Firefighter problem, we can show that
$z$ has no more than $L$ many $z$-loose vertices. 
Thus, we can first include all $z$-loose vertices
in the set $R$ of vertices to be protected by increasing
the budget of each level $\ell > h$ by at most
$L\leq 2^{h+1} \leq 2^\ell$.
The remaining vertices in $\supp(z)$ are well structures
(no two of them lie on the same leaf-root path), and an 
integral solution can be obtained easily.
The new budget value is $\bar{B}=2B+1$, where the ``$+1$''
term pays for the loose vertices.

The following theorem formalizes the above reasoning
and generalizes it in two ways. First, for a leaf $u\in \Gamma$
to be part of $W$, we required it to have a total $x$-value
of at least $0.5$ within the bottom levels; we will allow
for replacing $0.5$ by an arbitrary threshold $\mu\in (0,1]$.
Second, the level $h$ defining what is top and bottom
can be chosen to be of the form $h=\lfloor \log^{(q)} L\rfloor$
for $q\in \mathbb{Z}_{\geq 0}$, where
$\log^{(q)} L \coloneqq
\log\log\dots\log L$ is the value obtained by
taking $q$ many logs of $L$, and
by convention we set $\log^{(0)}L \coloneqq L$.
The generalization in terms of $h$ can be thought of as
iterating the above procedure on the RMFC instance
restricted to $V_{\leq h}$.

\begin{theorem}\label{thm:bottomCover}
Let $B\in \mathbb{R}_{\geq 1}$, $\mu \in (0,1]$, 
$q\in \mathbb{Z}_{\geq 1}$, and
$h = \lfloor \log^{(q)} L\rfloor$.
Let $x\in P_B$ with $\supp(x)\subseteq V_{> h}$,
and we define $W=\{u\in \Gamma \mid x(P_u) \geq \mu\}$.
Then one can efficiently compute
a set $R\subseteq V_{>h}$ such that
\smallskip
\begin{enumerate}[nosep,label=(\roman*)]
\item $R\cap P_u \neq \emptyset \quad \forall u\in W$, and
\item $\chi^R \in P_{B'}$, where $B'= \frac{q}{\mu}B + 1$
and $\chi^R\in \{0,1\}^{V\setminus \{r\}}$ is the
characteristic vector of $R$.
\end{enumerate}

\end{theorem}

Theorem~\ref{thm:bottomCover} has several interesting
consequences.
It immediately implies an
LP-based $O(\log^* N)$-approximation for RMFC, thus
matching the currently best approximation result
by Chalermsook and Chuzhoy~\cite{ChalermsookChuzhoy2010}:
It suffices to start with an optimal LP solution $B\geq 1$
and $x\in \bar{P}_B$ and invoke the above theorem with
$\mu=1$, $q=1+\log^* L$.
Notice that by definition of $\log^*$ we have
$\log^* L = \min \{\alpha \in \mathbb{Z}_{\geq 0} \mid
\log^{(\alpha)} L \leq 1\}$; hence
$h=\lfloor \log^{(1+\log^* L)} L\rfloor = 0$, implying that
all levels are bottom levels.
Since the integrality gap of the LP
is~$\Omega(\log^* N)=\Omega(\log^* L)$,
Theorem~\ref{thm:bottomCover} captures the limits of what
can be achieved by techniques based on the standard LP.

Interestingly, Theorem~\ref{thm:bottomCover} also implies
that the $\Omega (\log^* L)$ integrality gap is only
due to the top levels of the instance. More precisely, if,
for any $q=O(1)$ and $h=\lfloor \log^{(q)} L \rfloor$,
one would know what vertices an optimal solution $R^*$ protects
within the levels $V_{\leq h}$, then a constant-factor
approximation for RMFC follows easily 
by solving an LP on the
bottom levels $V_{> h}$ and using Theorem~\ref{thm:bottomCover}
with $\mu=1$
to round the obtained solution.

Also, using Theorem~\ref{thm:bottomCover} it is not hard
to find constant-factor approximation algorithms for RMFC
if the optimal budget $B_\OPT$ is large enough, say
$B \geq \log L$.\footnote{Actually, the argument we
present in the following works for any
$B = \log^{(O(1))}L$. However, we later only
need it for $B\geq \log L$ and thus focus 
on this case.}
The main idea is to solve the LP and define
$h=\lfloor \log L \rfloor$. Leaves that are largely
cut off by $x$ on bottom levels can be handled using
Theorem~\ref{thm:bottomCover}. For the remaining leaves,
which are cut off mostly on top levels, we can resolve an
LP only on the top levels $V_{\leq h}$ to cut them off.
This LP solution is sparse and contains at most $h\leq B$
loose nodes. Hence, all loose vertices can be selected
by increasing the budget by at most $h\leq B$, leading
to a well-structured residual problem for which one can
easily find an integral solution.
The following theorem summarizes this discussion.
A formal proof for Theorem~\ref{thm:bigBIsGood}
can be found in Section~\ref{sec:proofsRMFC}. 

\begin{theorem}\label{thm:bigBIsGood}
There is an efficient algorithm that computes a
feasible solution to a (compressed) instance of
RMFC with budget $B\leq 3 \cdot \max\{\log L, B_{\OPT}\}$.
\end{theorem}

\medskip

In what follows, we therefore assume $B_\OPT < \log L$
and present an efficient way to partially
enumerate vertices to be protected on top levels, 
leading to the claimed $O(1)$-approximation.

\subsubsection*{Partial enumeration algorithm}

Throughout our algorithm, we set 
\begin{equation*}
h=\lfloor\log^{(2)} L\rfloor
\end{equation*}
to be the threshold level defining top vertices $V_{\leq h}$
and bottom vertices $V_{> h}$.
Within our enumeration procedure we will solve LPs
where we explicitly include some vertex set
$A\subseteq V_{\leq h}$ to be part of the protected
vertices, and also exclude some set $D\subseteq V_{\leq h}$
from being protected. Our enumeration works by growing
the sets $A$ and $D$ throughout the algorithm.
We thus define the following LP for two disjoint
sets $A,D \subseteq V_{\leq h}$:
\begin{equation}\label{eq:lpRMFCAD}
\begin{array}{*2{>{\displaystyle}r}c*2{>{\displaystyle}l}}
\min & B    &     &     & \\
     & x    &\in  & \bar{P}_B & \\
     & B    &\geq &   1 & \\
     & x(u) &=    &   1 & \quad \forall u\in A\\
     & x(u) &=    &   0 & \quad \forall u\in D\enspace .\\
\end{array}\tag{$\mathrm{LP(A,D)}$}\labeltarget{eq:lpRMFCADtarget}
\end{equation}
Notice that~\ref{eq:lpRMFCAD} is indeed an LP even though the
definition of $\bar{P}_B$ depends on $B$ (but it does so linearly).

Throughout our enumeration procedure, the disjoint
sets $A, D \subseteq V_{\leq h}$ that we consider are
always such that for any $u\in A\cup D$, we have
$P_u\setminus\{u\} \subseteq D$. In other words, the vertices
$A\cup D \cup \{r\}$ form the vertex set of a subtree
of $G$ such that no root-leaf path contains two vertices
in $A$. We call a disjoint pair  of
sets $A,D\subseteq V_{\leq h}$ with this property
a \emph{clean pair}.

Before formally stating our enumeration procedure,
we briefly discuss the main idea behind it.
Let $\OPT\subseteq V\setminus \{r\}$ be an optimal solution
to our (compressed) RMFC instance corresponding to some
budget $B_{\OPT} \in \mathbb{Z}_{\geq 1}$. We assume without loss
of generality that $\OPT$ does not contain redundancies, i.e.,
there is precisely one vertex of $\OPT$ on each leaf-root
path.
Assume that we already guessed some clean pair
$A,D \subseteq V_{\leq h}$ of vertex sets to be
protected and not to be protected, respectively,
and that this guess is compatible with $\OPT$, i.e.,
$A\subseteq \OPT$ and $D\cap \OPT=\emptyset$.
Let $(x,B)$ be an optimal solution to~\ref{eq:lpRMFCAD}.
Because we assume that the sets $A$ and $D$ are compatible with
$\OPT$, we have $B\leq B_{\OPT}$ because
$(B_\OPT, \chi^\OPT)$ is feasible for \ref{eq:lpRMFCAD}. We define
\begin{equation*}
W_x = \left\{u\in \Gamma \;\middle\vert\;
  x(P_u \cap V_{> h}) \geq \frac{2}{3}\right\}
\end{equation*}
to be the set of leaves cut off from the root
by an $x$-load of at least $\mu=\frac{2}{3}$
within bottom levels.
For each $u\in \Gamma\setminus W_x$,
let $f_u\in V_{\leq h}$ be the vertex closest
to the root among all vertices in
$(P_u \cap V_{\leq h}) \setminus D$, and we define
\begin{equation}\label{eq:defFx}
F_x = \{f_u \mid u\in \Gamma\setminus W_x\} \setminus A.
\end{equation}
Notice that by definition, no two vertices of $F_x$ lie on
the same leaf-root path.
Furthermore, every leaf $u\in \Gamma\setminus W_x$
is part of the subtree
$T_f$ for precisely one $f\in F_x$.
The main motivation for considering $F_x$ is that to guess
vertices in top levels, we can show that it suffices
to focus on vertices
lying below some vertex in $F_x$, i.e., vertices
in the set $Q_x = V_{\leq h} \cap (\cup_{f\in F_x} T_{f})$.
To exemplify this, we first consider the special case
$\OPT\cap Q_x = \emptyset$, which will also play
a central role later in the analysis of our algorithm.
We show that for this case we can get an
$O(1)$-approximation to RMFC, even though we may only
have guessed a proper subset $A\subsetneq \OPT\cap V_{\leq h}$
of the $\OPT$-vertices within the top levels.

\begin{lemma}\label{lem:goodEnum}
Let $(A, D)$ be a clean pair of
vertices that is compatible with $\OPT$, i.e.,
$A\subseteq \OPT, D\cap \OPT = \emptyset$,
and let $x$ be an optimal solution
to~\ref{eq:lpRMFCAD}.
Moreover, let $(y,\bar{B})$ be an optimal solution to
\hyperlink{eq:lpRMFCADtarget}{$\mathrm{LP(A,V_{\leq h} \setminus A)}$}.
Then, if $\OPT\cap Q_x=\emptyset$, we have
$\bar{B}\leq \frac{5}{2} B_{\OPT}$.

Furthermore, if $\OPT\cap Q_x = \emptyset$,
by applying Theorem~\ref{thm:bottomCover}
to $y\wedge \chi^{V_{> h}}$ with $\mu=1$ and $q=2$, a set 
$R\subseteq V_{> h}$ is obtained such that
$R\cup A$ is a feasible solution to RMFC with respect
to the budget $6 \cdot B_{\OPT}$.\footnote{For two vectors
$a,b\in \mathbb{R}^n$ we denote by $a\wedge b\in \mathbb{R}^n$
the component-wise minimum of $a$ and $b$.}
\end{lemma}
\begin{proof}
Notice that $\OPT\cap Q_x=\emptyset$
implies that for each $u\in \Gamma \setminus W_x$,
we either have $A\cap P_u \neq \emptyset$ and thus
a vertex of $A$ cuts $u$ off from the root, or
the set $\OPT$ contains a vertex on $P_u \cap V_{>h}$.
Indeed, consider a leaf $u\in \Gamma \setminus W_x$
such that $A\cap P_u = \emptyset$.
Then
$\OPT\cap Q_x = \emptyset$ implies that no vertex
of $T_{f_u}\cap V_{\leq h}$ is part of $\OPT$.
Furthermore, $P_{f_u}\setminus T_{f_u} \subseteq D$
because $(A,D)$ is a clean pair and $f_u$ is the
topmost vertex on $P_u$ that is not in $D$.
Therefore, $\OPT \cap P_u \cap V_{\leq h} = \emptyset$,
and since $\OPT$ must contain a vertex in $P_u$, we must
have $\OPT\cap P_u \cap V_{>h}\neq \emptyset$.

However, this observation implies
that $z=\frac{3}{2}(x\wedge \chi^{V_{>h}})
+(\chi^{\OPT} \wedge \chi^{V_{>h}})+\chi^A$
satisfies
$z(P_u) \geq 1$ for all $u\in \Gamma$.
Moreover we have $z\in P_{\frac{3}{2}B+B_{\OPT}}$
due to the following.
First, $x\wedge \chi^{V_{>h}} \in P_B$ and
$\chi^{\OPT} \in P_{B_{\OPT}}$, which implies
$z-\chi^A\in P_{\frac{3}{2}B+B_{\OPT}}$.
Furthermore, $\chi^A\in P_B$, and the vertices in
$A$ are all on levels $V_{\leq h}$ which are disjoint
from the levels on which vertices in 
$\supp(z-\chi^A)\subseteq V_{>h}$ lie,
and thus do not compete
for the same budget.
Hence, $(z,\frac{3}{2}B+B_{\OPT})$ is feasible for
\hyperlink{eq:lpRMFCADtarget}{$\mathrm{LP(A,V_{\leq h} \setminus A)}$},
and thus
$\bar{B} \leq \frac{3}{2}B + B_{\OPT} \leq \frac{5}{2} B_{\OPT}$,
as claimed.

The second part of the lemma follows in a straightforward
way from Theorem~\ref{thm:bottomCover}.
Observe first that each leaf $u\in \Gamma$ is either
fully cut off from the root by $y$ on only top levels
or only bottom levels because $y$ is a $\{0,1\}$-solution
on the top levels $V_{\leq h}$, since on top levels it was
fixed to $\chi^A$ because it is a solution to
\hyperlink{eq:lpRMFCADtarget}{$\mathrm{LP(A,V_{\leq h} \setminus A)}$}.
Reusing the notation in Theorem~\ref{thm:bottomCover},
let $W=\{u\in \Gamma \mid (y\wedge \chi^{V_{> h}})(P_u) \geq 1\}$
be all leaves cut off from the root by $y\wedge \chi^{V_{>h}}$.
By the above discussion, every leaf is thus either part of $W$
or it is cut off from the root by vertices in $A$. 
Theorem~\ref{thm:bottomCover} guarantees that $R\subseteq V_{>h}$
cuts off all leaves in $W$ from the root, and hence, $R\cup A$
indeed cuts off all leaves from the root.
Moreover, by Theorem~\ref{thm:bottomCover}, the set
$R\subseteq V_{> h}$ is feasible with respect to the
budget $5B_{\OPT} +1 \leq 6 B_{\OPT}$.
Furthermore, $A$ is feasible for budget $B_{\OPT}$ because
it is a subset of $\OPT$. Since $A\subseteq V_{\leq h}$
and $R\subseteq V_{> h}$ are on disjoint levels, the
set $R\cup A$ is feasible for the budget $6 B_{\OPT}$.
\end{proof}

Our final algorithm is based on a recursive enumeration
procedure that computes a polynomial
collection of clean pairs $(A,D)$
such that there is one pair $(A,D)$ in the collection
with a corresponding LP solution $x$ of 
\hyperlink{eq:lpRMFCADtarget}{$\mathrm{LP(A,D)}$}
satisfying that the triple $(A,D,x)$ fulfills the conditions of
Lemma~\ref{lem:goodEnum}, and thus leading to a
constant-factor approximation.
Our enumeration algorithm
\hyperlink{alg:enumRMFCtarget}{$\mathrm{Enum}(A,D,\gamma)$}
is described below.
It contains a parameter $\gamma\in \mathbb{Z}_{\geq 0}$
that bounds the recursion depth of the enumerations.

\smallskip

{
\renewcommand{\thealgocf}{}
\begin{algorithm}[H]
\SetAlgorithmName{$\bm{\mathrm{Enum}(A,D,\gamma)}$
\labeltarget{alg:enumRMFCtarget}
}{}

\begin{enumerate}[rightmargin=1em]
\item Compute optimal solution $(x,B)$ to 
\hyperlink{eq:lpRMFCADtarget}{$\mathrm{LP(A,D)}$}.

\item\label{item:stopWhenBLarge}
\textbf{If} $B > \log L$\textbf{:} \textbf{stop}.
Otherwise, continue with step~\ref{item:addTriple}.

\item\label{item:addTriple}
Add $(A,D,x)$ to the family of triples to be considered.

\item\label{item:enumRecCall} \tm{if}{\textbf{I}}%
\textbf{f} $\gamma\neq 0$ \textbf{:}
\hfill \texttt{//recursion depth not yet reached \quad}

\quad \tm{for}{\textbf{F}}\textbf{or $u\in F_x$:}
\hfill \texttt{//$F_x$ is defined as in~\eqref{eq:defFx} \quad}

\quad\quad Recursive call to $\mathrm{Enum}(A\cup\{u\},D,\gamma-1)$.\\
\quad\quad \tm[overlay]{end}{}Recursive call
to $\mathrm{Enum}(A,D\cup \{u\},\gamma-1)$.

\begin{tikzpicture}[overlay, remember picture]
\draw (if) ++ (0,-0.5em) |- ($(if |- end) + (0.2,-0.2)$);
\draw (for) ++ (0,-0.5em) |- ($(for |- end) + (0.2,-0.1)$);
\end{tikzpicture}

\vspace{-1.5em}

\end{enumerate}

\caption{Enumerating triples $(A,D,x)$ to find one 
satisfying the conditions of Lemma~\ref{lem:goodEnum}.
}
\label{alg:enumRMFC}

\end{algorithm}
\addtocounter{algocf}{-1}
}%

\smallskip

Notice that for any clean pair $(A,D)$ and $u\in F_x$,
the two pairs $(A\cup \{u\}, D)$ and $(A, D\cup \{u\})$
are clean, too. Hence, if we start 
\hyperlink{alg:enumRMFCtarget}%
{$\mathrm{Enum}(A,D,\gamma)$}
with a clean pair $(A,D)$, we will encounter
only clean pairs during all recursive calls.

The key property of the above enumeration procedure
is that only a small recursion
depth $\gamma$ is needed for the enumeration algorithm 
to explore a good triple $(A,D,x)$, which satisfies
the conditions of Lemma~\ref{lem:goodEnum}, if we
start with the trivial clean pair $(\emptyset, \emptyset)$.
Furthermore, due to step~\ref{item:stopWhenBLarge},
we always have $B\leq \log L$ whenever the
algorithmm is in step~\ref{item:enumRecCall}. As we will see
later, this allows us to prove that $|F_x|$ is small, which
will limit the width of our recursive calls, and leads to
an efficient procedure as highlighted in the following Lemma.

\begin{lemma}\label{lem:enumWorks}
Let $\bar{\gamma}= 2(\log L)^2 \log^{(2)} L$.
The enumeration procedure \hyperlink{alg:enumRMFCtarget}%
{$\mathrm{Enum}(\emptyset,\emptyset,\bar{\gamma})$}
runs in polynomial time.
Furthermore, if $B_\OPT \leq \log L$, then
\hyperlink{alg:enumRMFCtarget}%
{$\mathrm{Enum}(\emptyset,\emptyset,\bar{\gamma})$} will
encounter a triple $(A,D,x)$ satisfying
the conditions of Lemma~\ref{lem:goodEnum}, i.e.,
\begin{enumerate}[nosep, label=(\roman*)]
\item $(A,D)$ is a clean pair,
\item $A\subseteq \OPT$,
\item $D\cap \OPT = \emptyset$, and
\item $\OPT\cap Q_x = \emptyset$.
\end{enumerate}
\end{lemma}

Hence, combining Lemma~\ref{lem:enumWorks} and
Lemma~\ref{lem:goodEnum} completes our enumeration procedure
and implies the following result.

\begin{corollary}\label{cor:summaryEnum}
Let $\mathcal{I}$ be an RMFC instance on $L$ levels
on a graph $G=(V,E)$ with budgets $B_\ell = 2^\ell \cdot B$.
Then there is a procedure with running time polynomial
in $2^L$, returning
a solution $(Q,B)$ for $\mathcal{I}$, where
$Q\subseteq V\setminus \{r\}$ is a set of vertices
to protect that is feasible for budget $B$,
satisfying the following:
If the optimal budget $B_{\OPT}$ for $\mathcal{I}$ satisfies
$B_{\OPT} \leq \log L$, then $B\leq 6 B_\OPT$.
\end{corollary}
\begin{proof}
It suffices to run 
\hyperlink{alg:enumRMFCtarget}%
{$\mathrm{Enum}(\emptyset,\emptyset,\bar{\gamma})$} to
first efficiently obtain a family of triples
$(A_i,D_i,x_i)_i$, where $(A_i, D_i)$ is a clean pair,
and $x_i$ is an optimal solution to
\hyperlink{eq:lpRMFCADtarget}{$\mathrm{LP(A_i,D_i)}$}.
By Lemma~\ref{lem:enumWorks}, one of these triples
satisfies the conditions of Lemma~\ref{lem:goodEnum}.
(Notice that these conditions cannot be checked since
it would require knowledge of $\OPT$.)
For each triple $(A_i,D_i,x_i)$ we obtain a corresponding
solution for $\mathcal{I}$ following the construction
described in Lemma~\ref{lem:goodEnum}. More precisely,
we first compute an optimal solution $(y_i,\bar{B}_i)$ to 
\hyperlink{eq:lpRMFCADtarget}{$\mathrm{LP(A_i,V_{\leq h} \setminus A_i)}$}.
Then, by applying Theorem~\ref{thm:bottomCover} to
$y_i\wedge \chi^{V_{> h}}$ with $\mu=1$ and $q=2$,
a set of vertices
$R_i\subseteq V_{> h}$ is obtained such that
$R_i\cup A_i$ is feasible for $\mathcal{I}$ for some
budget $B_i$.
Among all such sets $R_i\cup A_i$, we return the one
with minimum $B_i$.
Because Lemma~\ref{lem:enumWorks} guarantees that
one of the triples $(A_i, D_i, x_i)$ satisfies the
conditions of Lemma~\ref{lem:goodEnum}, we have by
Lemma~\ref{lem:goodEnum} that the best protection
set $Q=R_j\cup A_j$ among all $R_i\cup A_i$ has a
budget $B_j$ satisfying $B_j \leq 6 B_{\OPT}$.
\end{proof}

\subsection*{Summary of our $O(1)$-approximation for RMFC}

Starting with an RMFC instance $\mathcal{I}^{\mathrm{orig}}$
on a tree with $N$ vertices, we
first apply our compression result, Theorem~\ref{thm:compressionRMFC},
to obtain an RMFC instance $\mathcal{I}$ on a graph $G=(V,E)$ with depth
$L=O(\log N)$, and non-uniform budgets $B_\ell = 2^\ell B$
for $\ell\in [L]$.
Let $B_{\OPT}\in \mathbb{Z}_{\geq 1}$ be the optimal
budget value for $B$ for instance $\mathcal{I}$%
---recall that $B=B_{\OPT}$ in instance $\mathcal{I}$
implies that level $\ell\in [L]$ has budget $2^{\ell} \cdot B_{\OPT}$---%
and let $B_{\OPT}^{\mathrm{orig}}$
be the optimal budget for $\mathcal{I}^{\mathrm{orig}}$.
By Theorem~\ref{thm:compressionRMFC}, we have
$B_{\OPT} \leq B_{\OPT}^{\mathrm{orig}}$, and any solution
to $\mathcal{I}$ using budget $B$ can efficiently be transformed
into one of $\mathcal{I}^{\mathrm{orig}}$ of budget
$2B$.

We now invoke
Theorem~\ref{thm:bigBIsGood} and Corollary~\ref{cor:summaryEnum}.
Both guarantee that a solution to $\mathcal{I}$ with certain properties
can be computed efficiently.
Among the two solutions derived from Theorem~\ref{thm:bigBIsGood}
and Corollary~\ref{cor:summaryEnum}, we consider the one
$(Q,B)$ with lower budget $B$, where $Q\subseteq V\setminus \{r\}$
is a set of vertices to protect, feasible for budget
$B$.
If $B\geq \log L$, then Theorem~\ref{thm:bigBIsGood} implies
$B\leq 3 B_{\OPT}$, otherwise Corollary~\ref{cor:summaryEnum}
implies $B\leq 6 B_{\OPT}$. Hence, in any case we have
a $6$-approximation for $\mathcal{I}$. As mentioned before,
Theorem~\ref{thm:compressionRMFC} implies that the solution
$Q$ can efficiently be transformed into a solution for the
original instance $\mathcal{I}^{\mathrm{orig}}$ that is
feasible with respect to the budget
$2 B \leq 12 B_{\OPT} \leq 12 B^{\mathrm{orig}}_{\OPT}$,
thus implying Theorem~\ref{thm:O1RMFC}.

\section{Details on compression results}\label{sec:proofsCompression}

In this section, we present the proofs for our compression results,
Theorem~\ref{thm:compressionFF} and Theorem~\ref{thm:compressionRMFC}.
We start by proving Theorem~\ref{thm:compressionFF}. The same ideas are used
with a slight adaptation in the proof of Theorem~\ref{thm:compressionRMFC}. 

We call an instance $\overline{\mathcal{I}}$ obtained from 
an instance $\mathcal{I}$ by a sequence of down-push operations
a \emph{push-down of} $\mathcal{I}$.
We prove Theorem~\ref{thm:compressionFF} by proving
the following result, of which Theorem~\ref{thm:compressionFF}
is an immediate consequence, as we will soon show.

\begin{theorem}\label{thm:compressionDownPush}
Let $\mathcal{I}$ be a unit-budget Firefighter instance
with depth $L$, and let $\delta\in (0,1)$.
Then one can efficiently construct a push-down
$\overline{\mathcal{I}}$
of $\mathcal{I}$ such that
\smallskip
\begin{enumerate}[nosep,label=(\roman*)]
\item\label{item:closeToOPT}
   $\val(\OPT(\overline{\mathcal{I}}))
  \geq (1-\delta)\val(\OPT(\mathcal{I}))$, and

\item
$\overline{\mathcal{I}}$ has nonzero budget
on only $O(\frac{\log L}{\delta})$ levels.
\end{enumerate}
\end{theorem}

Before we prove Theorem~\ref{thm:compressionDownPush}, we show
how it implies 
Theorem~\ref{thm:compressionFF}.

\begin{proof}[Proof of Theorem~\ref{thm:compressionFF}]
We start by showing 
how levels of zero budget can be removed 
through the following \emph{contraction operation}. 
Let $\ell \in \{2,\dots, L\}$ be a level whose budget
is zero. For each vertex
$u \in V_{\ell-1}$ we contract all edges from $u$
to its children and increase the
weight $w(u)$ of $u$ by the sum of the weights
of all of its children.
Formally, if $u$ has children $v_1, \dots, v_k\in V_\ell$,
the vertices $u,v_1, \dots, u_k$ are replaced by a single
vertex $z$ with weight $w(z) = w(u) + \sum_{i=1}^k w(v_i)$,
and $z$ is adjacent to the parent of $u$ and to all children
of $v_1,\dots, v_k$.
One can easily observe that this is an ``exact''
transformation in the sense that any solution before
the contraction remains one after contraction
and vice versa (when identifying the vertex $z$
in the contracted version with $v$);
moreover, solutions before and
after contraction have the same value.

Now, by first applying Theorem~\ref{thm:compressionDownPush}
and then applying the latter contraction operations level by
level to all levels
$\ell\in \{2,\dots, L\}$
with zero budget (in an arbitrary order),
we obtain an equivalent instance with the desired 
depth, thus satisfying the conditions of
Theorem~\ref{thm:compressionFF}.
\end{proof}

It remains to prove Theorem~\ref{thm:compressionDownPush}.

\begin{proof}[Proof of Theorem~\ref{thm:compressionDownPush}]
Consider a unit-budget Firefighter instance on a tree
$G=(V,E)$ with depth $L$.
The push-down $\overline{\mathcal{I}}$ that we construct
will have nonzero budgets precisely on the following
levels $\mathcal{L} \subseteq [L]$:
\begin{equation*}
\mathcal{L} = \left\{\left\lceil(1+\delta)^j\right\rceil
    \;\middle\vert\; j\in \left\{0,\dots,
        \left\lfloor\frac{\log L}{\log(1+\delta)}
           \right\rfloor\right\}\right\}
   \cup \{L\}.
\end{equation*}
For simplicity, let $\mathcal{L}= \{\ell_1,\dots, \ell_k\}$
with $1=\ell_1 < \ell_2 < \dots < \ell_k=L$.
Hence,
$k=O(\frac{\log L}{\log(1+\delta)})
  = O(\frac{\log L}{\delta})$. The push-down
$\overline{\mathcal{I}}$ is obtained by pushing
any budget on a level not in $\mathcal{L}$ down
to the next level in $\mathcal{L}$. Formally,
for $i\in [k]$, the budget $B_{\ell_i}$
at level $\ell_i$ is given by
$B_{\ell_i} = \ell_i - \ell_{i-1}$, where
we set $\ell_{0}=0$.
Moreover, $B_\ell=0$ for
$\ell\in [L]\setminus \mathcal{L}$.
Clearly, the instance $\overline{\mathcal{I}}$ can be
constructed efficiently. Furthermore, the number
of levels with nonzero budget is equal to
$k=O(\frac{\log L}{\delta})$ as desired. It remains
to show point~\ref{item:closeToOPT}
of Theorem~\ref{thm:compressionDownPush}.

To show~\ref{item:closeToOPT}, consider an optimal
redundancy-free solution $S^*\subseteq V$ of $\mathcal{I}$; hence,
$\val(\OPT(\mathcal{I})) = \sum_{u\in S^*} w(T_u)$ and
no two vertices of $S^*$ lie on the same leaf-root path.
We will show that there is a feasible solution
$\overline{S}$ to $\overline{\mathcal{I}}$ such that
$\overline{S}\subseteq S^*$ and the value of
$\overline{S}$ is at least $(1-\delta)\val(\OPT(\mathcal{I}))$.
Notice that since $S^*$ is redundancy-free, any subset
of $S^*$ is also redundancy-free. Hence, the value of
the set $\overline{S}$ to construct will be equal
to $\sum_{u\in \overline{S}} w(T_u)$.
The set $S^*$ being (budget-)feasible for $\mathcal{I}$
implies 
\begin{equation}\label{eq:SStarFeasible}
|S^*\cap V_{\leq \ell}| \leq \ell
   \quad \forall \ell\in [L].
\end{equation}
Analogously, a set $S\subseteq V$ is feasible for
$\overline{\mathcal{I}}$ if and only if
\begin{equation}\label{eq:SFeasibleFull}
|S\cap V_{\leq \ell}| \leq \sum_{i=1}^\ell B_i
   \quad \forall \ell\in [L].
\end{equation}
Hence, we want to show that there is a set $\overline{S}$
satisfying the above system and such that
$\sum_{u\in \overline{S}}w(T_u)
  \geq (1-\delta)\val(\OPT(\mathcal{I}))$.
Notice that in~\eqref{eq:SFeasibleFull}, the constraint
for any $\ell\in [L-1]$ such that $B_{l+1}=0$ is
redundant due to the constraint for level $\ell+1$
which has the same right-hand side but a larger
left-hand side.
Thus, system~\eqref{eq:SFeasibleFull} is equivalent
to the following system
\begin{equation}\label{eq:SFeasibleShort}
\begin{aligned}
|S\cap V_{\leq \ell_{i+1}-1}| &\leq \ell_{i} 
   \quad \forall i\in [k-1],\\
|S\cap V| &\leq L.
\end{aligned}
\end{equation}
To show that there is a good subset
$\overline{S}\subseteq S^*$ that
satisfies~\eqref{eq:SFeasibleShort} we use a
polyhedral approach.
Observe that~\eqref{eq:SFeasibleFull} is the
constraint system of a laminar matroid
(see~\cite[Volume B]{Schrijver2003} for more information on matroids).
Hence,
the convex hull of all characteristic vectors
$\chi^S\in \{0,1\}^V$ of sets $S\subseteq S^*$
satisfying~\eqref{eq:SFeasibleShort} is given
by the following polytope
\begin{equation*}
P = \left\{
x\in [0,1]^V \;\middle\vert\;
\begin{minipage}[c]{0.4\linewidth}
\vspace{-1em}
\begin{align*}
x(V_{\leq \ell_{i+1}-1}) &\leq \ell_{i} \;\;\forall i\in [k-1],\\
x(V) &\leq L,\\
x(V\setminus S^*) &= 0
\end{align*}
\end{minipage}
\right\}.
\end{equation*}
Alternatively, to see that $P$ indeed
describes the correct polytope,
without relying on matroids, one can observe that its
constraint matrix is totally unimodular because it
has the consecutive-ones property with respect to the
columns.

Thus there exists a set $\overline{S}\subseteq S^*$ with
$\sum_{u\in \overline{S}} w(T_u) \geq (1-\delta)\val(\OPT(\mathcal{I}))$
if and only if
\begin{equation}\label{eq:polSb}
\max\left\{\sum_{u\in S^*} x(u)\cdot
 w(T_u) \;\middle\vert\;
    x\in P\right\}\geq (1-\delta)\val(\OPT(\mathcal{I})).
\end{equation}
To show~\eqref{eq:polSb}, and thus complete the proof,
we show that $y=\frac{1}{1+\delta} \chi^{S^*}\in P$.
This will indeed imply~\eqref{eq:polSb} since the
objective value of $y$ satisfies
\begin{equation*}
\sum_{u\in S^*} y(u) \cdot w(T_u) =
 \frac{1}{1+\delta}\val(\OPT(\mathcal{I}))
    \geq (1-\delta)\val(\OPT(\mathcal{I})).
\end{equation*}

To see that $y\in P$, notice that
$y(V\setminus S^*)=0$ and
$y(V) = \frac{1}{1+\delta} |S^*|
\leq \frac{1}{1+\delta} L \leq L$, where the
first inequality follows by $S^*$
satisfying~\eqref{eq:SStarFeasible} for $\ell=L$.
Finally, for $i\in [k-1]$, we have
\begin{align*}
y(V_{\leq \ell_{i+1}-1}) &=
  \frac{1}{1+\delta}
    |S^* \cap V_{\leq \ell_{i+1}-1}|
\leq \frac{1}{1+\delta}(\ell_{i+1}-1),
\end{align*}
where the inequality follows from $S^*$
satisfying~\eqref{eq:SStarFeasible}
for $\ell=\ell_{i+1}-1$.
It remains to show $\ell_{i+1} -1 \leq (1+\delta)\ell_i$
to prove $y\in P$.
Let $\alpha \in \mathbb{Z}_{\geq 0}$ be the smallest
integer for which we have
$\ell_{i+1} = \lceil (1+\delta)^{\alpha}\rceil$. In
particular, this implies
$\ell_{i}=\lceil (1+\delta)^{\alpha-1}\rceil$. We
thus obtain
\begin{equation*}
\ell_{i+1} - 1 \leq (1+\delta)^{\alpha}
   = (1+\delta) (1+\delta)^{\alpha-1}
   \leq (1+\delta) \ell_i,
\end{equation*}
as desired.

\end{proof}

We conclude with the proof of Theorem~\ref{thm:compressionRMFC}.

\begin{proof}[Proof of Theorem~\ref{thm:compressionRMFC}.]

 We start by describing the construction of $G' = (V',E')$. As is the case
in the proof of Theorem~\ref{thm:compressionFF}, we first change the 
budget assignment of the instance and then contract all levels with zero budgets.
Notice that, for a given budget $B$ per layer,
we can consider an RMFC instance as a Firefighter instance,
where each leaf $u\in \Gamma$ has weight $w(u)=1$, and all other
weights are zero. Since our goal is to save all leaves, we want
to save vertices of total weight $|\Gamma|$.

For simplicity of presentation we assume that $L$ is a power of $2$. This assumption does
not compromise generality, as one can always augment the original tree with one path starting from the root and going down to level
$2^{\lceil\log L\rceil}$.

The set of levels in which the transformed instance will have
nonzero budget is 
\begin{equation*}
\mathcal{L} = \left\{2^j-1 \,\middle\vert\, j\in \{1,\ldots, \log L \} \right\}.
\end{equation*}
However, instead of down-pushes we will do \emph{up-pushes} were
budget is moved upwards. More precisely, 
the budget of any level $\ell\in [L]\setminus \mathcal{L}$
will be assigned to the first level in $\mathcal{L}$ that
is above $\ell$, i.e., has a smaller index than $\ell$.
As for the Firefighter case, we now remove all $0$-budget
levels using contraction, which will lead to a new
weight function $w'$ on the vertices. Since our goal
is to save the weight of the whole tree,
we can remove for each vertex $u$ with $w'(u) > 0$, the
subtree below $u$. This does not change the problem since
we have to save $u$, and thus will anyway also save its subtree.
This finishes our construction of $G'=(V',E')$, and the task
is again to remove all leaves of $G'$.
Notice that $G'$ has $L' \leq |\mathcal{L}| = \log L $
many levels, and level $\ell\in [L']$ has a budget of
$B 2^{\ell}$ as desired.
Analogous to the
discussion for compression in the context of the Firefighter 
problem we have that if the original problem is feasible,
then so is the RMFC problem on $G'$ with
budgets $B 2^{\ell}$.
Indeed, before performing the contraction operations (which
do not change the problem), the original RMFC problem was
a push-down of the one we constructed.

Similarly, one can observe that before contraction,
the instance we obtained is itself a push-down of
the original instance with budgets $2B$ on each level.
Hence, analogously to the compression result for
the Firefighter case, any solution to the RMFC problem
on $G'$ can 
efficiently be transformed into a solution to the original
RMFC problem on $G$ with budgets $2B$ on each level.

\end{proof}

\section{Missing details for Firefighter PTAS}\label{sec:proofsFF}

In this section we present the missing proofs for our PTAS for the
Firefighter problem.

We start by proving Lemma~\ref{lem:sparsityFF}, showing that
any vertex solution $x$ to \ref{eq:lpFF} has
few $x$-loose vertices.
More precisely, the proof below shows that the number
of $x$-loose vertices is upper bounded by the number
of tight budget constraints.
The precise same reasoning used in the proof of
Lemma~\ref{lem:sparsityFF} can also be applied
in further contexts, in particular for the RMFC problem.

\subsubsection*{Proof of Lemma~\ref{lem:sparsityFF}}

Let $x$ be a vertex of the polytope defining the feasible set
of~\ref{eq:lpFF}.
Hence, $x$ is uniquely defined by
$|V\setminus\{r\}|$ many linearly independent and tight
constraints of this polytope.
Notice that the tight constraints can be partitioned into
three groups:
\begin{enumerate}[label=(\roman*),nosep]
\item Tight nonnegativity constraints, one for
each vertex in $\mathcal{F}_1=\{u\in V\setminus \{r\} \mid x(u) = 0\}$.

\item Tight budget constraints, one for each level in
$\mathcal{F}_2 = \{\ell\in [L] \mid x(V_{\leq \ell})=\sum_{i=1}^\ell B_i\}$.

\item Tight leaf constraints, one for each vertex in
$\mathcal{F}_3 = \{u\in \Gamma \mid x(P_u) = 1\}$.
\end{enumerate}
Due to potential degeneracies of the polytope describing
the feasible set of~\ref{eq:lpFF} there may be several
options to describe $x$ as the unique solution to
a full-rank linear subsystem of the constraints described
by $\mathcal{F}_1 \cup \mathcal{F}_2 \cup \mathcal{F}_3$.
We consider a system that contains all tight
nonnegativity constraints, i.e.,
constraints corresponding to $\mathcal{F}_1$, and
complement these constraints with arbitrary subsets
$\mathcal{F}'_2\subseteq \mathcal{F}_2$ and 
$\mathcal{F}'_3\subseteq \mathcal{F}_3$ of
budget and leaf constraints that lead to a full rank
linear system corresponding to the constraints
$\mathcal{F}_1 \cup \mathcal{F}'_2 \cup \mathcal{F}'_3$.
Hence
\begin{equation}\label{eq:fullRankSys}
|\mathcal{F}_1| + |\mathcal{F}'_2| + |\mathcal{F}'_3| = |V| - 1.
\end{equation}

Let $V^{\mathcal{L}}\subseteq \supp(x)$
and $V^{\mathcal{T}}\subseteq \supp(x)$
be the $x$-loose and $x$-tight vertices, respectively.
We first show $|\mathcal{F}'_3|\leq |V^{\mathcal{T}}|$.
For each leaf $u\in \mathcal{F}'_3$, let $f_u\in V^\mathcal{T}$ be 
the first vertex on the unique $u$-root path that is part of
$\supp(x)$. In particular, if $u\in \supp(x)$ then $f_u=u$.
Clearly, $f_u$ must be an $x$-tight vertex because
the path constraint with respect to $u$ is tight.
Notice that for any distinct vertices $u_1,u_2\in \mathcal{F}'_3$,
we must have $f_{u_1}\neq f_{u_2}$. Assume by sake of
contradiction that $f_{u_1}= f_{u_2}$. However, this implies
$\chi^{P_{u_1}} - \chi^{P_{u_2}}\in \spn(\{\chi^{v} \mid v\in \mathcal{F}_1\})$, since 
$P_{u_1} \Delta P_{u_2} := (P_{u_1} \setminus P_{u_2})\cup (P_{u_2} \setminus P_{u_1}) \subseteq \mathcal{F}_1$, and leads to a contradiction
because we exhibited a linear dependence among the constraints
corresponding to $\mathcal{F}'_3$ and $\mathcal{F}_1$.
Hence, $f_{u_1}\neq f_{u_2}$ which implies that the
map $u \mapsto f_u$ from $\mathcal{F}'_3$ to $V^{\mathcal{T}}$
is injective and thus
\begin{equation}\label{eq:boundLeafConstr}
|\mathcal{F}'_3| \leq |V^{\mathcal{T}}|.
\end{equation}
We thus obtain
\begin{align*}
|\supp(x)| &= |V|-1-|\mathcal{F}_1|
    && \text{($\supp(x)$ consists of all $u\in V\setminus \{r\}$ with
               $x(u)\neq 0$, i.e., $u\not\in \mathcal{F}_1$)}\\
  &= |\mathcal{F}'_2| + |\mathcal{F}'_3|
    && \text{(by~\eqref{eq:fullRankSys})}\\
  &\leq |\mathcal{F}'_2| + |V^{\mathcal{T}}|
    && \text{(by~\eqref{eq:boundLeafConstr})},
\end{align*}
which leads to the desired result since
\begin{equation*}
|V^{\mathcal{L}}| = |\supp(x)| - |V^{\mathcal{T}}|
  \leq |\mathcal{F}'_2| \leq L.
\end{equation*}

\qed

\subsubsection*{Proof of Lemma~\ref{lem:pruning}}
Within this proof we focus on protection sets where the budget available
for any level is spent on the same level (and not a later one).
As discussed, there is always an optimal protection set
with this property.

Let $B_\ell \in \mathbb{Z}_{\geq 0}$ be the budget available at level $\ell\in [L]$ and let 
$\lambda_\ell = \lambda B_\ell$.
 We construct the tree $G'$ using the following greedy procedure. Process
the levels of $G$ from the first one to the last one. At every level $\ell\in [L]$,
pick $\lambda_\ell$ vertices $u^\ell_1, \cdots, u^\ell_{\lambda_\ell}$ at the $\ell$-th 
level of $G$ greedily, i.e., pick each next vertex such that the subtree corresponding to that 
vertex has largest weight among all remaining vertices in the level. 
After each selection of a vertex the greedy procedure can no longer 
select any vertex in the corresponding subtree in subsequent iterations.\footnote{
For $\lambda=1$ this procedure produces a set of vertices, which comprise
a $\frac{1}{2}$-approximation for the Firefighter problem, as it coincides
with the greedy algorithm of Hartnell and Li~\cite{HartnellLi2000}.}

Now, the tree $G'$ is constructed by deleting from $G$ any vertex
that is both not contained in any subtree $T_{u^\ell_i}$, and not 
contained in any path $P_{u^\ell_i}$ for $\ell\in [L]$ and $i\in [\lambda_\ell]$.
In other words, if $U\subseteq V$ is the set of all leaves
of $G$ that were disconnected from the root by the greedy
algorithm, then we consider the subtree of $G$ induced
by the vertices $\cup_{u\in U}P_u$.
Finally, the weights of vertices on the paths 
$P_{u^\ell_i} \setminus \{u^\ell_i\}$ for $\ell\in [L]$ and $i\in [\lambda_\ell]$ are reduced
to zero. This concludes the construction of $G'=(V',E')$ and the new weight function $w'$. Denote
by $D_\ell = \{u^\ell_1,\cdots, u^\ell_{\lambda_\ell}\}$ the set of vertices chosen by the
greedy procedure in level $\ell$, and let $D=\cup_{\ell\in [L]} D_{\ell}$.
Observe that by construction we have that each vertex
with non-zero weight is in the subtree of a vertex in $D$, i.e.,
$$
w'(V') = \sum_{u\in D} w'(T'_u).
$$
The latter immediately implies point~\ref{item:pruningLargeOpt}
of Lemma~\ref{lem:pruning} because the vertices $D$ can
be partitioned into $\lambda$ many vertex sets that are
budget-feasible and can thus be protected in a Firefighter solution.
Hence an optimal solution to the Firefighter problem
on $G'$ covers at least a $\frac{1}{\lambda}$-fraction of the total
weight of $G'$.

It remains to prove point~\ref{item:pruningSmallLoss} of the Lemma.
Let $S^* = S^*_1\cup \cdots \cup S^*_L$ be the vertices protected in some optimal
solution in $G$, where $S^*_\ell \subseteq V_\ell$ are the vertices protected in level $\ell$ (and
hence $|S^*_\ell| \leq B_\ell$). 
Without loss of generality, we assume that $S^*$ is redundancy-free.
For distinct vertices $u,v\in V$ we say that $u$ \emph{covers} $v$ if $v\in T_u \setminus \{u\}$.

For $\ell \in [L]$, let $I_\ell = S^*_l \cap D_\ell$ be the set of vertices protected 
by the optimal solution that are also chosen by the greedy algorithm in level $\ell$.
Furthermore, let $J_\ell \subseteq S^*_\ell$
be the set of vertices of the optimal solution that are 
covered by vertices chosen by the greedy algorithm in earlier
iterations, i.e.,
$J_\ell = S^*_\ell \cap \bigcup_{u\in D_1\cup\cdots\cup D_{\ell -1}} T_u$. 
Finally, let $K_\ell = S^*_\ell \setminus (I_\ell \cup J_\ell)$ be all other
optimal vertices in level $\ell$. Clearly, $S^*_\ell = I_\ell \cup J_\ell \cup K_\ell$ 
is a partition of $S^*_\ell$.

Consider a vertex $u\in K_\ell$ for some $\ell\in [L]$. From the guarantee of the greedy 
algorithm it holds that for every vertex $v\in D_\ell$ we have $w'(T_v) = w(T_v) \geq w(T_u)$. 
The same does not necessarily hold for covered vertices. 
On the other hand, covered vertices
are contained in $G'$ with their original weights. We exploit these two 
properties to prove the existence of a solution in $G'$
of almost the same weight as $S^*$.

To prove the existence of a good solution we construct
a solution $A = A_1 \cup \cdots \cup A_L$ with $A_\ell \subseteq V_\ell$ and $|A_\ell| \leq B_\ell$
randomly, and prove a bound on its expected quality.
We process the levels of the tree $G'$ top-down to construct $A$ step
by step.
This clearly does not compromise generality. Recall that we only need to prove the 
existence of a good solution, and not compute it efficiently. We can hence assume the
knowledge of $S^*$ in the construction of $A$. To this end assume that all levels
$\ell' < \ell$ were already processed, and the corresponding sets $A_{\ell'}$ were
constructed. The set $A_{\ell}$ is constructed as follows:

\begin{enumerate}
\item Include in $A_\ell$ all vertices in $I_\ell$.
\item Include in $A_\ell$ all vertices in $J_\ell$ that are not 
covered by vertices in $A_1\cup \cdots \cup A_{\ell-1}$ (vertices selected so far).
\item Include in $A_\ell$ a \emph{uniformly random subset} of $|K_\ell|$ vertices
from $D_\ell \setminus I_\ell$.
\end{enumerate}

It is easy to verify that the latter algorithm returns a redundancy-free solution, as no two
chosen vertices in $A$ lie on the same path to the root. Next, we show that the expected
weight of vertices saved by $A$ is at least $(1-\frac{1}{\lambda})\val(\OPT(\overline{\mathcal{I}}))$, 
which will prove our claim, since then at least one solution has the desired quality.

Since we only need a bound on the expectation we can focus on a single level $\ell \in [L]$ 
and show that the contribution of vertices in $A_\ell$ is in expectation at least $1-\frac{1}{\lambda}$
times the contribution of the vertices in $S^*_\ell$. Observe that the vertices in $I_\ell$ are
contained both in $S^*_\ell$ and in $A_\ell$, hence it suffices to show that the contribution
of $A_\ell \setminus I_\ell$ is at least $1-\frac{1}{\lambda}$ times the contribution 
of $S^*_\ell \setminus I_\ell$, in expectation. Also, recall that every vertex in $D_\ell$
contributes at least as much as any vertex in $K_\ell$, by the greedy selection rule. It follows
that the $|K_\ell|$ randomly selected vertices in $A_\ell$ have at least as much contribution
as the vertices in $K_\ell$. Consequently, to prove the claim is suffices to bound the 
expected contribution of vertices in $A_\ell \cap J_\ell$ with respect to the contribution of
$J_\ell$. Since $A_\ell \cap J_\ell \subseteq J_\ell$ it suffices to show that every vertex
$u\in J_\ell$ is also present in $A_\ell$ with probability at least $1-\frac{1}{\lambda}$.

To bound the latter probability we make use of the random choices in the construction
of $A$ as follows. Let $\ell' < \ell$ be the level at which for some $w\in D_{\ell'}$ it 
holds that $u\in T_w$. In other words, $\ell'$ is the level that contains the ancestor 
of $u$ that was chosen by the greedy construction of $G'$. Now, since $S^*$ is redundancy-free,
and by the way that $A$ is constructed, it holds that if $u\not\in A_\ell$ 
then $w\in A_{\ell'}$, namely if $u$ is covered, it can only be covered by the 
unique ancestor $w$ of $u$ that was chosen in the greedy construction of $G'$. Furthermore,
in such a case the vertex $w$ was selected randomly in the third step of the $\ell'$-th
iteration. Put differently, the probability that the vertex $u$ is covered 
is exactly the probability that its ancestor $w$ is chosen randomly to be part of $A_{\ell'}$.
Since these vertices are chosen to be a random subset of $|K_{\ell'}|$ vertices from the set $D_{\ell'}\setminus I_{\ell'}$,
this probability is at most 
$$
\frac{|K_{\ell'}|}{|D_{\ell'}| - |I_{\ell'}|} =
\frac{|K_{\ell'}|}{\lambda B_{\ell'} - |I_{\ell'}|} \leq 
\frac{1}{\lambda}, 
$$
where the last inequality follows from $|K_{\ell'}| + |I_{\ell'}| \leq B_{\ell'}$.
This implies that $u\in A_\ell$ with probability of at least $1-\frac{1}{\lambda}$, as required
and concludes the proof of the lemma.

\qed

\subsubsection*{Proof of Lemma~\ref{lem:setQ}}

We construct the set $Q$ in two phases as follows. First we construct 
a set $\overline Q \subseteq H$ of vertices fulfilling the first and the third properties, i.e.,
it will satisfy $|\overline Q| = O(\frac{\log N}{\epsilon^3})$, as well as the property that
$G[V\setminus \overline Q\cup \{r\}]$ has connected components each of weight at most $\eta$. Then,
we add to $\overline Q$ all vertices of $H$ of degree at least three to arrive
at the final set $Q$.

It will be convenient to define heavy vertices and heavy tree with respect to any 
subtree $G'= (V', E')$ of $G$ which contains the root $r$. Concretely, we 
define $H_{G'} = \{u\in V'\setminus \{r\} \,\mid\, w(T'_u)\geq \eta\}$ 
to be the set of $G'$-heavy vertices. The $G'$-heavy tree is the
subtree $G'[H_{G'} \cup \{r\}]$ of $G'$. Observe that $H = H_G$ and that
$H_{G'} \subseteq H$ for every subtree $G'$ of $G$.

To construct $\overline Q$ we process the tree $G$ in a bottom-up 
fashion starting with $\overline Q = \emptyset$. We will also remove
parts of the tree in the end of every iteration. The first iteration 
starts with $G' = G$. In every iteration that starts with tree $G'$, include in 
$\overline Q$ an arbitrary leaf $u\in H_{G'}$ of the heavy tree and remove $u$ and all vertices
in its subtree from $G'$. The procedure ends when there is
either no heavy vertex in $G'$ anymore, or when $r$ is the
only heavy vertex in $G'$.

Let us verify that the claimed properties indeed hold. The fact that 
$|\overline Q| = O(\frac{\log N}{\epsilon^3})$ follows from the fact that at each iteration 
we remove a $G'$-heavy vertex including all its subtree from the 
current tree $G'$. This implies that the total weight of the tree $G'$
decreases by at least $\eta$ in every iteration. Since we only include one 
vertex in every iteration we have
$|\overline Q| \leq \frac{w(V)}{\eta} = O(\frac{\log N}{\epsilon^3})$.

The third property follows from the fact that we always remove a leaf
of the $G'$-heavy tree. Observe that the connected components of 
$G[V\setminus (\overline Q \cup \{r\})]$ are contained in the subtrees
we disconnect in every iteration in the construction of $\overline Q$.
By definition of $G'$-heavy leaves, in any such iteration where 
a $G'$-heavy leaf $u$ is removed from the tree, these parts have weight 
at least $\eta$, but any subtree rooted at any descendant of $u$ has
weight strictly smaller than $\eta$ (otherwise this descendant would
be $G'$-heavy as well, contradicting the assumption that it has a
$G'$-heavy leaf $u$ as an ancestor). Now, since $u$ is included in $\overline Q$,
the connected components are exactly these subtrees, so the property indeed holds.

To construct $Q$ and conclude the proof it remains to include in $\overline Q$
all remaining nodes of degree at least three in the heavy tree. The 
fact that also all leaves of the heavy tree are included in $Q$ is
readily implied by the construction of $\overline Q$, so the second property 
holds for $Q$. Clearly, by removing more vertices from the heavy tree, the sizes
of connected components only get smaller, so $Q$ also satisfies the third
condition, since $\overline Q$ already did. Finally, the number of 
vertices of degree at least three in the heavy tree is strictly
less than the number of its leaves, which is $O(\frac{\log N}{\epsilon^3})$;
for otherwise a contradiction would occur since the tree would
have an average degree of at least $2$.
This implies that, in
total, $|Q| = O(\frac{\log N}{\epsilon^3})$,
so the first property also holds.

To conclude the proof of the lemma it remains to note that the latter
construction can be easily implemented in polynomial time.

\qed

\section{Missing details for $O(1)$-approximation
for RMFC}\label{sec:proofsRMFC}

This section contains the missing proofs for our
$12$-approximation for RMFC.

\subsection*{Proof of Theorem~\ref{thm:bottomCover}}

To prove Theorem~\ref{thm:bottomCover} we first show
the following result, based on which Theorem~\ref{thm:bottomCover}
follows quite directly.

\begin{lemma}\label{lem:sliceCover}
Let $B\in \mathbb{R}_{\geq 1}$, $\eta\in (0,1]$,
$k \in \mathbb{Z}_{\geq 1}$, and
$\ell_1 = \lfloor \log^{(k)} L \rfloor$,
$\ell_2 = \lfloor \log^{(k-1)} L \rfloor$.
Let $x\in P_B$ with
$\supp(x)\subseteq V_{(\ell_1,\ell_2]}
 \coloneqq V_{>\ell_1} \cap V_{\leq \ell_2}$,
and we define $Y = \{u\in \Gamma \mid x(P_u) \geq \eta\}$.
Then one can efficiently compute a
set $R\subseteq V_{(\ell_1,\ell_2]}$ such
that
\smallskip
\begin{enumerate}[nosep, label=(\roman*)]
\item\label{item:scHitPath}
$R\cap P_u \neq \emptyset \quad \forall u\in Y$, and

\item\label{item:scBudgetOk}
$\chi^R\in P_{\bar{B}}$,
where $\bar{B} = \frac{1}{\eta} B + 1$.
\end{enumerate}
\end{lemma}

We first observe that Lemma~\ref{lem:sliceCover} indeed
implies Theorem~\ref{thm:bottomCover}.

\begin{proof}[Proof of Theorem~\ref{thm:bottomCover}]
For $k=1,\dots, q$, let
$\ell_1^k = \lfloor \log^{(k)} L\rfloor$ and
$\ell_2^k = \lfloor \log^{(k-1)} L\rfloor$, and we define
$x^k\in P_B$ by $x^k = x \wedge \chi^{V_{(\ell_1^k, \ell_2^k]}}$.
Hence, $x=\sum_{k=1}^q x^k$.
For each $k\in [q]$, we apply Lemma~\ref{lem:sliceCover} to
$x^k$ with $\eta = \frac{\mu}{q}$ to obtain a set
$R^k \subseteq V_{(\ell_1^k, \ell_2^k]}$ satisfying
\begin{enumerate}[nosep, label=(\roman*)]
\item $R^{k}\cap P_u \neq \emptyset$ \quad
$\forall u\in Y^k=\{u\in \Gamma \mid x^k(P_u) \geq \eta\}$, and

\item $\chi^{R^k} \in P_{\bar{B}}$, where
$\bar{B} \coloneqq \frac{1}{\eta} B + 1 = \frac{q}{\mu} B + 1
  \eqqcolon B'$.
\end{enumerate}
We claim that $R=\cup_{k=1}^q R^k$ is a set satisfying
the conditions of Theorem~\ref{thm:bottomCover}.
The set $R$ clearly satisfies $\chi^R \in P_{B'}$
since $\chi^{R^k}\in P_{B'}$ for $k\in [q]$
and the sets $R^k$ are on disjoint levels.
Furthermore, for each $u\in W=\{v\in \Gamma \mid x(P_v)\geq \mu\}$
we indeed have $P_u\cap R\neq\emptyset$ due to the following.
Since $x=\sum_{k=1}^q x^k$ and $x(P_u) \geq \mu$ there exists
an index $j\in [q]$ such that $x^j(P_u) \geq \eta = \frac{\mu}{q}$,
and hence $P_u \cap R \supseteq P_u \cap R^j \neq \emptyset$.

\end{proof}

Thus, it remains to prove Lemma~\ref{lem:sliceCover}.

\begin{proof}[Proof of Lemma~\ref{lem:sliceCover}]\leavevmode

Let $\tilde{B} = \frac{1}{\eta} B$.
We start by determining an optimal vertex solution $y$
to the linear program $\min\{z(V\setminus \{r\}) \mid z\in Q\}$,
where
\begin{equation*}
Q = \{z\in P_{\tilde{B}}
  \mid z(u) = 0\;\forall u\in V
    \setminus (V_{(\ell_1,\ell_2]} \cup \{r\}),\;\;
z(P_u) \geq 1 \;\forall u\in Y\}.
\end{equation*}
Notice that $Q\neq \emptyset$
since $\frac{1}{\eta} x \in Q$; hence, the above
LP is feasible.
Furthermore, notice that $y(P_u)\leq 1$ for $u\in \Gamma$;
for otherwise, there is a vertex $v\in \supp(y)$ such that
$y(P_v) > 1$, and hence $y - \epsilon \chi^{\{v\}}\in Q$
for a small enough $\epsilon >0$, violating that 
$y$ is an \emph{optimal} vertex solution.

Let $V^{\mathcal{L}}$ be all $y$-loose vertices.
We will show that the set
\begin{equation*}
R = V^{\mathcal{L}} \cup \{u\in V\setminus \{r\} \mid y(u)=1\}
\end{equation*}
fulfills the properties claimed by the lemma.
Clearly, $R\subseteq V_{(\ell_1,\ell_2]}$ since
$\supp(y) \subseteq V_{(\ell_1,\ell_2]}$.

To see that condition~\ref{item:scHitPath}
holds, let $u\in Y$, and notice that we have $y(P_u)=1$.
Either $|P_u \cap \supp(y)| =1$, in which case
the single vertex $v$ in $P_u\cap \supp(y)$ satisfies
$y(u)=1$ and is thus contained in $R$; or $|P_u\cap \supp(y)| > 1$,
in which case $P_u\cap V^{\mathcal{L}} \neq \emptyset$ which again
implies $R\cap P_u \neq \emptyset$.

To show that $R$ satisfies~\ref{item:scBudgetOk},
we have to show that $R$ does not exceed the budget
$\bar{B}\cdot 2^\ell = (\frac{1}{\eta}B + 1) 2^\ell$ of any
level $\ell\in \{\ell_1+1,\dots, \ell_2\}$.
We have
\begin{align*}
|R\cap V_\ell| \leq y(V_\ell) + |V^{\mathcal{L}}|
\leq \tilde{B} 2^\ell + |V^{\mathcal{L}}|
= \frac{1}{\eta} B 2^\ell + |V^{\mathcal{L}}|,
\end{align*}
where the second inequality follows from $y\in Q$.
To complete the proof it suffices to show
$|V^{\mathcal{L}}| \leq 2^\ell$.
This follows by a sparsity reasoning analogous to
Lemma~\ref{lem:sparsityFF} implying that the number
of $y$-loose vertices is bounded by the number
of tight budget constraints, and thus
\begin{equation}\label{eq:budgetBoundsFirstStep}
|V^{\mathcal{L}}| \leq \ell_2 - \ell_1 \leq \ell_2
 = \lfloor \log^{(k-1)} L \rfloor.
\end{equation}
Furthermore,
\begin{align*}
2^\ell &\geq 2^{\ell_1+1} = 2^{\lfloor \log^{(k)} L \rfloor + 1}
\geq 2^{\log^{(k)} L} = \log^{(k-1)} L,
\end{align*}
which, together with~\eqref{eq:budgetBoundsFirstStep},
implies $|V^{\mathcal{L}}| \leq 2^\ell$ and thus
completes the proof.

\end{proof}

\subsection*{Proof of Theorem~\ref{thm:bigBIsGood}}

Let $(y,B)$ be an optimal solution to the RMFC
relaxation $\min\{B \mid x\in \bar{P}_B\}$
and let $h=\lfloor \log L \rfloor$.
Hence, $B\leq B_\OPT$.
We invoke Theorem~\ref{thm:bottomCover} with respect
to the vector $y\wedge \chi^{V_{>h}}$ and $\mu=0.5$
to obtain a set $R_1\subseteq V_{>h}$ satisfying
\begin{enumerate}[nosep,label=(\roman*)]
\item $R_1\cap P_u\neq \emptyset
\quad\forall u\in W$,
and

\item $\chi^{R_1} \in P_{2B+1}$,
\end{enumerate}
where
$W = \{u\in \Gamma \mid y(P_u\cap V_{>h}) \geq 0.5\}$.
Hence, $R_1$ cuts off all leaves in $W$ from
the root by only protecting vertices on
levels $V_{> h}$ and using budget bounded by
$2B+1 \leq 3B \leq 3 \max\{\log L, B_\OPT\}$.

We now focus on the leaves $\Gamma \setminus W$,
which we will cut off from the root by protecting
a vertex set $R_2 \subseteq V_{\leq h}$ feasible
for budget $3 \max\{\log L, B_\OPT\}$.
Let $(z,\bar{B})$ be an optimal vertex
solution to the
following linear program
\begin{equation}\label{eq:reoptTop}
\min\left\{\bar{B} \;\middle\vert\;
x\in P_{\bar{B}},\; 
x(P_u) = 1 \;\forall u\in \Gamma\setminus W
\right\}.
\end{equation}
First, notice that~\eqref{eq:reoptTop} is feasible
for $\bar{B}\leq 2B$. This follows by observing
that the vector $q= 2(y\wedge \chi^{V_{\leq h}})$
satisfies $q\in P_{2 B}$ since $y\in P_B$.
Moreover, for $u\in \Gamma\setminus W$,
we have
\begin{equation*}
q(P_u) = 2 y(P_u \cap V_{\leq h})
= 2 (1-y(P_u \cap V_{> h})) > 1,
\end{equation*}
where the last inequality follows from
$y(P_u\cap V_{>h}) < 0.5$ because
$u\in \Gamma\setminus W$.
Finally, there exists a vector
$q' < q$ such that
$q'(P_u) =1$ for $u\in \Gamma\setminus W$.
The vector $q'$ can be obtained from $q$ by
successively reducing values on vertices
$v\in \supp(q)$ satisfying
$q(P_v) > 1$.
This shows that $(q',2B)$ is a feasible
solution to~\eqref{eq:reoptTop} and hence
$\bar{B} \leq 2B$.

Consider the set of all $z$-loose vertices
$V^{\mathcal{L}}=\{u\in \supp(z) \mid z(P_u)<1\}$.
We define
\begin{equation*}
R_2 = V^{\mathcal{L}} \cup 
\{u\in \supp(z) \mid z(u)=1\}.
\end{equation*}
Notice that for each $u\in \Gamma\setminus W$,
the set $R_2$ contains a vertex on the
path from $u$ to the root. Indeed, either
$|\supp(z)\cap P_u|=1$ in which case there
is a vertex $v\in P_u$ with $z(v)=1$, which is
thus contained in $R_2$, or $|\supp(z)\cap P_u|>1$
in which case the vertex $v\in \supp(z)\cap P_u$
that is closest to the root among all vertices in
$\supp(z)\cap P_u$ is a $z$-loose vertex.
Hence, the set $R=R_1\cup R_2$ cuts off all leaves
from the root. It remains to show that it is
feasible for budget $3 \max\{\log L, B_\OPT\}$.

Using an analogous sparsity reasoning as in
Lemma~\ref{lem:sparsityFF}, we obtain that
$|V^{\mathcal{L}}|$ is bounded by the number
of tight budget constraints, which is at most
$h=\lfloor \log L \rfloor \leq \log L$.
Hence, for any level $\ell\in [h]$, we have
\begin{align*}
|R_2 \cap V_\ell| &\leq |V^{\mathcal{L}}| + z(V_\ell) \\
  &\leq \log L + 2^\ell \bar{B} && \text{($(z,\bar{B})$ feasible
for~\eqref{eq:reoptTop})}\\
  &\leq \log L + 2^\ell \cdot (2 B) && \text{($\bar{B}\leq 2B$)}\\
  &\leq 2^\ell \cdot (3 \max\{\log L, B_\OPT\}).
     && \text{($B\leq B_\OPT$)}
\end{align*}
Thus, both $R_1$ and $R_2$ are budget-feasible for
budget $3 \max\{\log L, B_\OPT\}$, and since they
contain vertices on disjoint levels, $R=R_1\cup R_2$
is feasible for the same budget.

\qed

\subsection*{Proof of Lemma~\ref{lem:enumWorks}}

To show that the running time of
\hyperlink{alg:enumRMFCtarget}%
{$\mathrm{Enum}(\emptyset,\emptyset,\bar{\gamma})$}
is polynomial, we show that there is only a polynomial
number of recursive calls to
\hyperlink{alg:enumRMFCtarget}%
{$\mathrm{Enum}(A,D,\gamma)$}. Notice that the number
of recursive calls done in one execution of
step~\ref{item:enumRecCall} of the algorithm is equal
to $2 |F_x|$.
We thus start by upper bounding $|F_x|$ for any solution
$(x,B)$ to \ref{eq:lpRMFCAD} with $B < \log L$.
Consider a vertex $f_u\in F_x$, where
$u\in \Gamma\setminus W_x$.
Since $u$ is a leaf not in $W_x$, we have
$x(P_u \cap V_{\leq h}) > \frac{1}{3}$, and
thus
\begin{equation*}
x(T_{f_u}\cap V_{\leq h}) > \frac{1}{3}
\quad \forall f_u \in F_x.
\end{equation*}
Because no two vertices of $F_x$ lie on the same
leaf-root path, the sets $T_{f_u} \cap V_{\leq h}$
are all disjoint for different $f_u\in F_x$,
and hence
\begin{align*}
\frac{1}{3}|F_x| &< \sum_{f \in F_x} x(T_{f}\cap V_{\leq h})\\
 &\leq x(V_{\leq h})
    && \text{(disjointness of sets $T_{f}\cap V_{\leq h}$
              for different $f \in F_x$})\\
 &\leq \sum_{\ell=1}^h 2^\ell B
    && \text{($x$ satisfies budget constraints of~\ref{eq:lpRMFCAD} )}\\
 &< 2^{h+1} B\\
 &< 2 (\log L)^2.
    && \text{($h=\lfloor \log^{(2)} L \rfloor$ and $B < \log L$)}
\end{align*}
Since the recursion depth is
$\bar{\gamma}=2(\log L)^2 \log^{(2)} L$,
the number of recursive calls is bounded by
\begin{align*}
O\left((2 |F_x|)^{\bar{\gamma}}\right) &= 
(\log L)^{O((\log L)^2 \log^{(2)} L)}
=2^{o(L)} = o(N),
\end{align*}
thus showing that
\hyperlink{alg:enumRMFCtarget}%
{$\mathrm{Enum}(\emptyset,\emptyset,\bar{\gamma})$}
runs in polynomial time.

It remains to show that \hyperlink{alg:enumRMFCtarget}%
{$\mathrm{Enum}(\emptyset,\emptyset,\bar{\gamma})$} finds
a triple satisfying the conditions of Lemma~\ref{lem:goodEnum}.
For this we identify a particular execution path of the
recursive procedure 
\hyperlink{alg:enumRMFCtarget}%
{$\mathrm{Enum}(\emptyset,\emptyset,\bar{\gamma})$} that,
at any point in the algorithm, will maintain a clean
pair $(A,D)$ that is compatible with $\OPT$,
i.e., $A\subseteq \OPT$ and $D\cap \OPT = \emptyset$.
At the beginning of the algorithm we clearly have
compatibility with $\OPT$ since $A=D=\emptyset$.
To identify the execution path we are interested
in, we highlight which recursive call we want to follow
given that we are on the execution path.
Hence, consider a clean pair $(A,D)$
that is compatible with $\OPT$ and assume we are
within the execution of
\hyperlink{alg:enumRMFCtarget}%
{$\mathrm{Enum}(A,D,\gamma)$}.
Let $(x,B)$ be an optimal solution to~\ref{eq:lpRMFCAD}.
Notice that $B \leq B_\OPT \leq \log L$, because
$(A,D)$ is compatible with $\OPT$.
If $\OPT\cap Q_x=\emptyset$, then $(A,D,x)$ fulfills the
conditions of Lemma~\ref{lem:goodEnum} and we are done.
Hence, assume $\OPT\cap Q_x \neq \emptyset$, and
let $f \in F_x$ be such that
$\OPT\cap T_{f}\cap V_{\leq h}\neq \emptyset$.
If $f \in \OPT$, then consider the execution path
continuing with the call of 
\hyperlink{alg:enumRMFCtarget}%
{$\mathrm{Enum}(A\cup \{f\},D,\gamma-1)$}; otherwise,
if $f\not\in \OPT$, we focus on the call of
\hyperlink{alg:enumRMFCtarget}%
{$\mathrm{Enum}(A,D\cup \{f\},\gamma-1)$}.
Notice that compatibility with $\OPT$ is maintained
in both cases.

To show that the thus identified execution path of 
\hyperlink{alg:enumRMFCtarget}%
{$\mathrm{Enum}(\emptyset,\emptyset,\bar{\gamma})$}
indeed leads to a triple satisfying the conditions
of Lemma~\ref{lem:goodEnum}, we measure progress
as follows. For any clean pair $(A,D)$
compatible with $\OPT$, we define
a potential function $\Phi(A,D)\in \mathbb{Z}_{\geq 0}$
in the following way.
For each $u\in \OPT\cap V_{\leq h}$,
let $d_u\in \mathbb{Z}_{\geq 0}$
be the distance of $u$ to the first vertex in
$A\cup D \cup \{r\}$ when following the unique
$u$-$r$ path. We define
$\Phi(A,D)= \sum_{u\in \OPT \cap V_{\leq h}} d_u$.
Notice that as long as we have a triple $(A,D,x)$
on our execution path that does
not satisfy the conditions of Lemma~\ref{lem:goodEnum},
then the next triple $(A',D',x')$ on our execution
path satisfies $\Phi(A',D') < \Phi(A,D)$.
Hence, either we will encounter a triple on our
execution path satisfying
the conditions of Lemma~\ref{lem:goodEnum}
while still having a strictly positive potential,
or we will encounter a triple $(A,D,x)$ compatible
with $\OPT$ and $\Phi(A,D)=0$, which implies
$\OPT\cap V_{\leq h} = A$,
and we thus correctly guessed all vertices of
$\OPT\cap V_{\leq h}$ implying that
the conditions of Lemma~\ref{lem:goodEnum}
are satisfied for the triple $(A,D,x)$.
Since $\Phi(A,D)\geq 0$ for any compatible clean
pair $(A,D)$, this implies that a triple
satisfying the conditions of Lemma~\ref{lem:goodEnum}
will be encountered if the recursion depth $\bar{\gamma}$
is at least $\Phi(\emptyset,\emptyset)$.
To evaluate $\Phi(\emptyset,\emptyset)$ we have to compute
the sum of the distances of all
vertices $u\in \OPT\cap V_{\leq h}$
to the root. The distance of $u$ to the root is at
most $h$ since $u\in V_{\leq h}$. Moreover, 
$|\OPT \cap V_{\leq h}| < 2^{h+1} B_{\OPT}$
due to the budget constraints. Hence,
\begin{align*}
\Phi(\emptyset, \emptyset)
  &< h \cdot 2^{h+1} \cdot B_{\OPT}\\
  &\leq 2 \log^{(2)} L \cdot (\log L)^2
    && \text{($h=\lfloor \log^{(2)} L \rfloor$ and $B_\OPT \leq \log L$)}\\
  &= \bar{\gamma},
\end{align*}
implying that a triple fulfilling the conditions of
Lemma~\ref{lem:goodEnum} is encountered by
\hyperlink{alg:enumRMFCtarget}%
{$\mathrm{Enum}(\emptyset,\emptyset,\bar{\gamma})$}.

\qed

\section*{Acknowledgements}
We are grateful to Noy Rotbart for many stimulating discussions
and for bringing several relevant references to our attention.

\appendix

\section{Basic transformations for the Firefighter problem}\label{apx:trans}

In this section we provide some basic transformations showing how
different natural variations of the Firefighter problem can be reduced
to each other.
We start by proving Lemma~\ref{lem:genBudgetsToUnit}.

\begin{proof}[Proof of Lemma~\ref{lem:genBudgetsToUnit}]
 Consider an instance of the weighted Firefighter problem with general budgets 
consisting of a tree $G = (V,E)$ of depth $L$ rooted at the vertex $r\in V$, weights 
$w(u) \in \mathbb{Z}_{\geq 0}$ for all $u\in V\setminus \{r\}$
and budgets $B_\ell \in \mathbb{Z}_{> 0}$
for all $\ell\in [L]$. We transform the instance into an equivalent instance with 
unit budgets by performing the following simple
steps for all levels $V_\ell$ for $\ell\in [L]$:
\begin{itemize}
 \item For every $u\in V_\ell$, subdivide the edge connecting $u$ to its ancestor
in $G$ into a path with $B_\ell$ edges, by introducing $B_\ell-1$ new vertices. Denote
the nodes on this path, excluding the ancestor of $u$ in $G$, by $Y_u$.

 \item Set the weight of all new vertices to zero, while maintaining the weight $w(u)$ 
for the original vertex $u$.
\end{itemize}

Denote the resulting tree by $G' = (V', E')$.
To conclude the construction it remains to allow one unit of budget in every level of
the transformed tree.
It is easy to verify that feasible solutions to the Firefighter problem for 
the two instances are in correspondence. A feasible solution for $G$ is transformed
to a solution in $G'$ by replacing the $B_\ell$
vertices $S_\ell$ protected in any level $V_\ell$ of $G$
with any $B_\ell$ vertices on the corresponding paths $\{Y_u \,\mid\, u\in S_\ell\}$ in $G'$, one in 
each of the $B_\ell$ distinct levels of $G'$ that are in correspondence with $V_\ell$. The opposite
transformation selects for every protected vertex $u\in V'$ in a feasible solution 
for $G'$ the vertex $u\in V$ such that $u'\in Y_u$. It is straightforward to verify that
in both transformations the obtained solutions are feasible and that they have 
weights identical to the original solutions.

Finally, since $B_\ell \leq n$ can be assumed for every $\ell\in [L]$, each one of the 
$n-1$ edges in $G$ is subdivided into a path of length at most $n$, thus the 
number of vertices in $G'$ is at most $O(n^2)$.

\end{proof}

We remark that a construction analogous to the one used
in the proof of Lemma~\ref{lem:genBudgetsToUnit} can be used
to show that RMFC with non-uniform budgets can be reduced to
the uniform budget case. In an RMFC instance with non-uniform
budgets,
the budget on level $\ell$ is equal to $B\cdot a_\ell$,
where $a_\ell \in \mathbb{Z}_{> 0}$ for $\ell\in [L]$ are
given as input, and the goal is still to find the minimum $B$
to protect vertices that cut off all leaves from the root
and fulfill the budget constraints.

\medskip

Next, we show how a weighted instance of the Firefighter problem can be
transformed into a unit-weight one with only an arbitrarily small loss
in term of the objective function.

\begin{lemma}\label{lem:generalToUnitWeights}
Let $\delta > 0$ and $\alpha \in (0,1]$. 
Any weighted unit-budget Firefighter problem on a
tree $G = (V,E)$ and weights $w(u)\in \mathbb{Z}_{\geq 0}$ 
for $u\in V\setminus \{r\}$ can be transformed efficiently
into a polynomial-size unit-weight unit-budget
Firefighter problem on a tree $G' = (V', E')$
such that any $\alpha$-approximate feasible solution 
for $G'$ can be efficiently transformed
into a $(1-\delta)\alpha$-approximate solution for $G$.
\end{lemma}

\begin{proof}
Assume $w(V)>0$, i.e., not all weights are zero,
since for otherwise the result trivially holds.
Notice that this assumption also implies that
the value $\val(\OPT)$ of an optimal Firefighter solution
in $G$ satisfies $\val(\OPT)\geq 1$.

For simplicity we present the transformation in two steps,
each losing at most a $\frac{\delta}{2}$-fraction in terms of objective.
First we use a standard scaling 
and rounding technique to obtain a new weight function
that is bounded by a polynomial in the size of the tree.
Concretely, we construct weights
$w'(u) \in \mathbb{Z}_{\geq 0}$ for $u\in V\setminus \{r\}$
such that $w'(u) = O(\frac{n}{\delta})$ for every $u\in V$, and 
for a well-chosen parameter $D\in \mathbb{R}_{>0}$
we have: 
\begin{equation}\label{eq:scalingProps}
Dw'(S) \leq w(S) \leq Dw'(S) + \frac{\delta}{2} \val(\OPT)
\qquad \forall S\subseteq V\setminus \{r\}.
\end{equation}
In a second phase discussed below we use the obtained instance to
construct a unit-weight instance with the desired
property. 

Let $w_{\max} = \max_{u\in V\setminus \{r\}} w(u)$
be the maximum weight of any vertex in $G$.
Define $D = \delta w_{\max} /2n$, where $n=|V|$, and for
every $u\in V\setminus \{r\}$ set $w'(u) = \fl{w(u)/D}$.
Observe that $\val(\OPT) \geq w_{\max}$ since any single
vertex can be protected.
The latter 
scaling indeed fulfills the desired properties
as $w'(u)\leq 2n/\delta = O(n/\delta)$, and for every
$S\subseteq V\setminus \{r\}$ we have 
\begin{align*}
Dw'(S) &\leq w(S) \leq D w'(S) + D |S|
  \leq Dw'(S) + \frac{\delta}{2} \val(\OPT),
\end{align*}
where the first two inequalities follows from
$w'(u) = \lfloor w(u)/D\rfloor \;\forall u\in V\setminus \{r\}$,
and the last one
from $D |S| \leq D n =\delta w_{\max}/2 \leq \delta \val(\OPT)/2$.
This shows~\eqref{eq:scalingProps}.

We show next that the latter transformation loses
at most a $\frac{\delta}{2}$-fraction in
the objective function. More precisely,
let $S'\subseteq V\setminus \{r\}$
be a set of vertices that will not burn in an $\alpha$-approximate
solution to the Firefighter problem with respect to the
weights $w'$. We will show that
$w(S') \geq (1-\frac{\delta}{2}) \alpha \val(\OPT)$, implying
that the same solution is $(1-\frac{\delta}{2}) \alpha$-approximate
with respect to the original weights $w$.
Let $S^* \subseteq V\setminus \{r\}$ be the vertices
that will not burn at the end of the process
in an optimal solution for $G$.
By~\eqref{eq:scalingProps} we have
$D w'(S^*) + \frac{\delta}{2} \val(\OPT) \geq w(S^*) = \val(\OPT)$,
implying $Dw'(S^*) \geq (1-\frac{\delta}{2})\val(\OPT)$.
We conclude:
\begin{align*}
\left(1-\frac{\delta}{2}\right)\val(\OPT) &\leq D w'(S^*) \\
   &\leq \frac{1}{\alpha} D w'(S')
      && \text{($S'$ corresponds to an $\alpha$-approximate solution
                for weights $w'$)}    \\
   &\leq \frac{1}{\alpha} w(S'),
      && \text{(since $w'(u) = \lfloor{w(u)}/{D}\rfloor$
                 for $u\in V\setminus \{r\}$)}
\end{align*}
which yields
$w(S') \geq (1-\frac{\delta}{2})\alpha \val(\OPT)$, as desired. 

\smallskip

Next we present the second transformation, which,
given a weighted Firefighter problem
with tree $G=(V,E)$ and integer
weights $w(u)\in \mathbb{Z}_{\geq 0}$ bounded by $O(n)$, transforms it into a 
unit-weight instance on a new tree $G'= (V',E')$
by losing at most a $\frac{\delta}{2}$-fraction 
in terms of objective.

The tree $G'$ is obtained from $G$ by taking a
copy of $G$ and attaching
$\lfloor\frac{4n}{\alpha \delta}w(u)\rfloor$ new
leaves to every vertex $u\in V\setminus \{r\}$. 
For brevity, for a vertex set
$R\subseteq V\setminus \{r\}$, we denote by
$\mathrm{sv}(R)\subseteq V$ the set of all
vertices that will not burn in $G$ if one protects
the set $R$,
i..e, $\mathrm{sv}(R)=\cup_{u\in R} T_u$.
Similarly, for $R'\subseteq V'\setminus \{r\}$,
we denote by $\mathrm{sv}'(R')=\cup_{u\in R'} T'_u \subseteq V'$
all vertices in $G'$ that will not burn if $R'$ gets protected.

Consider a solution that protects
a set $R'\subseteq V'\setminus \{r\}$ of vertices in $G'$. 
Observe that $V \cap R'$ is a feasible set of
vertices to protect in $G$.
We can upper bound the objetive value of $R'$ in $G'$ as
follows, where $w'$ is the unit-weight function used in $G'$:
\begin{equation}\label{eq:svPrimeLeqSv}
\begin{aligned}
w'(\mathrm{sv}'(R')) &= |\mathrm{sv}'(R')|
  \leq |R'\setminus V|
     + \sum_{u \in \mathrm{sv}(R'\cap V)}
       \left(1+\left\lfloor
           \frac{4n}{\alpha \delta} w(u)\right\rfloor\right)\\
  &\leq 
   n + \sum_{u\in \mathrm{sv}(R'\cap V)}
          \left( 1 + \frac{4n}{\alpha\delta} w(u) \right)\\
  &\leq 2n + \frac{4n}{\alpha\delta}w(\mathrm{sv}(R'\cap V)).
\end{aligned}
\end{equation}
Moreover, for any 
set of vertices $R\subseteq V\setminus \{r\}$ 
in $G$ we have
\begin{equation}\label{eq:svLeqSvPrime}
\begin{aligned}
w'(\mathrm{sv}'(R)) &= |\mathrm{sv}'(R)|
  = \sum_{u\in \mathrm{sv}(R)}\left( 1 +
        \left\lfloor\frac{4 n}{\alpha \delta} w(u) \right\rfloor \right)\\
  &\geq \sum_{u\in \mathrm{sv}(R)} \frac{4n}{\alpha\delta} w(u)\\
  &= \frac{4 n}{\alpha \delta} w(\mathrm{sv}(R)).
\end{aligned}
\end{equation}
We complete the rest of the proof similar to the
proof of the first transformation.
Let $R^* \subseteq V\setminus \{r\}$ be
an optimal set of vertices to protect in $G$, and let
$R'\subseteq V'\setminus \{r\}$ be an $\alpha$-approximation
for the unit-weight Firefighter instance on $G'$.
Our goal is to show that $R'\cap V$ is a solution
to the Firefighter problem on $G$ of value at least
$(1-\frac{\delta}{2}) \alpha \val(\OPT)$. Indeed,
we have
\begin{align*}
\val(\OPT) &= w(\mathrm{sv}(R^*))
    \leq \frac{\alpha \delta}{4 n} w'(\mathrm{sv}'(R^*))
        && \text{(by~\eqref{eq:svLeqSvPrime})}\\
  &\leq \frac{\delta}{4n} w'(\mathrm{sv}'(R'))
        && \text{(since $R'$ is an $\alpha$-approximation for $G'$)}\\
  &\leq \frac{\delta}{2} + \frac{1}{\alpha} w(\mathrm{sv}(R'\cap V))
        && \text{(by~\eqref{eq:svPrimeLeqSv})}\\
  &\leq \frac{\delta}{2}\val(\OPT)
        + \frac{1}{\alpha} w(\mathrm{sv}(R'\cap V)),
      && \text{(because $\val(\OPT)\geq 1$)}
\end{align*}
which implies
\begin{equation*}
w(\mathrm{sv}(R'\cap V))
  \geq \left(1-\frac{\delta}{2}\right)\alpha \val(\OPT),
\end{equation*}
as desired.

Finally, both transformations can be implemented in polynomial time. For the
first transformation this is trivial, while for the second transformation one uses
the fact that the input weights are polynomially bounded, and hence $G'$ has
polynomial size.
\end{proof}

\end{document}